\documentclass[12pt, draftclsnofoot,onecolumn]{IEEEtran}

 \usepackage{blindtext, graphicx}
 \usepackage{epstopdf}
 \usepackage{amsmath}
 \usepackage{amsfonts}
 \usepackage{amssymb}
  \usepackage{subcaption}
 \usepackage{tikz}
 \usetikzlibrary{arrows}
\usepackage{eqnarray,amsmath}

\usepackage{blindtext}
\usepackage{graphicx}
\usepackage{multirow}
\usepackage{url}
\usepackage[T1]{fontenc}
\usepackage{mathptmx}
\usepackage{amsthm}
\usepackage{dsfont}% From Table~\ref{Tb:computational_complexity} it can be observed that the computational complexity of LN-1, LN-1E, and Grad-Desc vary linearly with the number of anchor nodes in the network.  

\usepackage{amssymb}
\usepackage{algorithm,algorithmic}
\usepackage{caption,subcaption}
\newtheorem{thm}{Theorem}[section]
\newtheorem{lem}[thm]{Lemma}

\usepackage{pifont}% http://ctan.org/pkg/pifont
\usepackage{soul}
\theoremstyle{remark}
\newtheorem*{remark}{Remark}
\usepackage[labelformat=simple]{subcaption}

\usepackage{cite}
\ifCLASSINFOpdf
\else
\fi

\newcommand\blankfootnote[1]{%
  \let\thefootnote\relax\footnotetext{#1}%
  \let\thefootnote\svthefootnote%
}

\DeclareMathAlphabet{\mathcal}{OMS}{cmsy}{m}{n}
\SetMathAlphabet{\mathcal}{bold}{OMS}{cmsy}{b}{n}
\newcommand{\cmark}{\ding{51}}%
\newcommand{\xmark}{\ding{55}}%

\usepackage{eso-pic}
\newcommand\AtPageUpperMyleft[1]{\AtTextUpperLeft{
 \put(\LenToUnit{0\paperwidth},\LenToUnit{1.2cm}){
     \parbox{\textwidth}{\raggedright\fontsize{12}{9}\selectfont #1}}
 }}

\newcommand{\conf}[1]{
\AddToShipoutPictureBG*{
\AtPageUpperMyleft{#1}
}
}

\begin{document}
% \pagestyle{empty} % don't print page number
% \begin{center}
%   \large\bfseries  % font for your notice
%  This work has been submitted to the IEEE for possible publication. Copyright may be transferred
% without notice, after which this version may no longer be accessible
% \end{center}
% \setcounter{page}{1} % page number for start of manuscript
% \pagestyle{headings}% per your code

%\onecolumn % switch to one column

%\twocolumn % switch to 2 columns
%\setcounter{page}{1} % page number for start of manuscript
%\pagestyle[headings} % per your code

\title{RSSI-based Secure Localization in the Presence of Malicious Nodes in Sensor Networks}

\conf{This work has been submitted to the IEEE for possible publication. Copyright may be transferred
without notice, after which this version may no longer be accessible.} % Change according to their suggestion

\author{Bodhibrata~Mukhopadhyay,
        Seshan~Srirangarajan,
        and~Subrat~Kar% <-this % stops a space
\thanks{B. Mukhopadhyay, S. Srirangarajan and  S. Kar  are with the Department of Electrical Engineering, Indian Institute of Technology Delhi, New Delhi, India. S. Srirangarajan and  S. Kar are also associated with the Bharti School of Telecommunication Technology and Management, Indian Institute of Technology Delhi. B.~Mukhopadhyay is supported through the Visvesvaraya PhD Scheme Fellowship from the Ministry of Electronics and Information Technology (MEITY), Govt. of India. (e-mail:  bodhibrata@gmail.com, seshan@ee.iitd.ac.in, subrat@ee.iitd.ac.in).}
}

\maketitle

\begin{abstract}
The ability of a sensor node to determine its location in a sensor network is important in many applications. The infrastructure for the location-based services is an easy target for malicious attacks. We address scenarios where malicious node(s) attempt to disrupt, in an uncoordinated or coordinated manner, the localization process of a target node. We propose four techniques for secure localization: weighted least square (WLS), secure weighted least square (SWLS), and $\ell_1$-norm based techniques LN-1 and LN-1E, in a network that includes one or more compromised anchor nodes. WLS and SWLS techniques are shown to offer significant advantage over existing techniques by assigning larger weights to the anchor nodes that are closer to the target node, and by detecting the malicious nodes and eliminating their measurements from the localization process. In a coordinated attack, the localization problem can be posed as a plane fitting problem where the measurements from non-malicious and malicious anchor nodes lie on two different planes. LN-1E technique estimates these two planes and prevents disruption of the localization process. The Cramer-Rao lower bound (CRLB) for the position estimate is also derived. The proposed techniques are shown to provide better localization accuracy than the existing algorithms.
\end{abstract}

\begin{IEEEkeywords}
\noindent Received signal strength (RSS), uncoordinated attack, coordinated attack, Cramer\--Rao lower bound (CRLB), least square (LS), secure localization.
\end{IEEEkeywords}

% \begin{keyword}
% %% keywords here, in the form: keyword \sep keyword
% Pyroelectric Infrared Sensor \sep indoor localization \sep regression \sep multilateration \sep ranging
% %% PACS codes here, in the form: \PACS code \sep code

% %% MSC codes here, in the form: \MSC code \sep code
% %% or \MSC[2008] code \sep code (2000 is the default)
% \end{keyword}

%\end{frontmatter}



\ifpdf
    \graphicspath{{pictures/}}
\else
    \graphicspath{{pictures/}}
\fi
\section{Introduction}
Rapid developments in the field of micro-electronics, integrated circuit fabrication, and embedded software have increased the computational power, lifetime, and sensing capabilities of wireless sensor nodes. A wireless sensor network (WSN) is formed by a collection of authenticated sensor nodes that communicate among themselves and cooperate for a common purpose. These nodes collect data and transmit it to the base station for further processing. As the monitoring area grows larger, the number of sensor nodes in a network increases and it becomes difficult to keep track of the locations of the sensor nodes manually. However, in the absence of accurate location information, the data from the sensors are not very useful. Localization methods based on different techniques have been presented in the literature including those based on optimization~\cite{localization_SDP_jiang,Salari_cooperative_loc,wang_loc_cooperative}, graph~\cite{graph_based,shen_locating_2016}, game theory~\cite{lee_distributed_2017}, least squares~\cite{kay1993fundamentals,Rappaport:2001:WCP:559977}, and fingerprinting~\cite{Jang_localization_survey,survey_fingerprint}.

Localization techniques are broadly divided into two categories: range-based techniques~\cite{ncc2015_bodhi,seshan_rssi,A_Survey_on_TOA,AoA_localization_with_RSSI_differences_of_directional_Jiang}, and range-free techniques~\cite{Range_free_localization_schemes_for_large_scalee_He,range_freelocalization_Stoleru1}. For determining the position of the target node, range-based techniques use the distance between the target node (node whose location is not known at the time of deployment) and the anchors (nodes whose locations are known at the time of deployment), whereas range-free techniques only use the connectivity information between the nodes. Scenarios where location-based services are used include monitoring activities and movement patterns in farm animals~\cite{WSN_cattle_sikka}, metropolitan air quality monitoring~\cite{air_quality_monitoring}, monitoring toddlers and elderly persons, tracking goods in the supply chain industry~\cite{Costa2013}, land slide detection~\cite{landslide_monitering}, and navigation tool for people in places such as shopping malls and airports.

These location-based services are not immune to security threats including cryptographic attacks such as unauthorized access or modification of the information. Attackers can also gain unauthorized access to the infrastructure responsible for carrying out the localization and/or modify the code on the sensor nodes and turn them into malicious nodes (non-cryptographic attacks). A malicious node does not provide correct information required for localization to the target node. Both types of attacks can result in erroneous location estimation of the target nodes. For determining its location, the target node requires certain information from its neighbouring nodes namely, received signal strength information~\cite{NCC2014_bodhi}, time of flight (for time of arrival (ToA) or time difference of arrival (TDoA)-based techniques), or connectivity (for techniques such as DV-Hop~\cite{DVhop_bao}) information. Various types of attacks on localization techniques have been studied in the literature~\cite{location_authen_gonzalez} such as impersonation (malicious node masquerading as an honest node), distance fraud (malicious node reporting information resulting in incorrect distance estimation such as by arbitrarily varying the transmit power level), time fraud (malicious nodes including incorrect time stamp into the packets sent to the target node), and Sybil attack (malicious node claiming multiple identities/locations representing multiple nodes). 

A non-malicious node can also seem to behave like a malicious node if the wireless link between this node and the target node is affected. This can happen if the direct path between the anchor and target node is obstructed, or the transmit antenna is damaged resulting in unusual variation in RSSI values. 
%Dil et al.~\cite{antennaei_orien_dil} experimentally showed that the RSS value can shift over 10 dB due to change in orientation of the antennae. 
In addition, distance estimation can be affected due to environmental changes as the path loss exponent depends on temperature and humidity~\cite{Boano_temp_impact,Baccour_2012}. In such scenarios the localization process can be adversely affected.
%
%In a recent experiment, %researchers under the guidance of professor Todd Humphreys from the department of aerospace engineering and engineering mechanics at the university of Texas (Austin) 
%snooping the GPS (Global Positioning System) signal allowed the change of the course of a super yacht in the Ionian Sea~\cite{gps_spoofing_Kugler}. A GPS device calculates its distance using the signals it receives from positioning satellites orbiting the earth. A GPS signal simulator can generate GPS signals like the positioning satellite. So, these fake signals can be transmitted to mislead the GPS receiver. A fake GPS signal can also be used to manipulate the boundary of a geofenced~\cite{geofencing} area. 
%
A GPS signal simulator can initiate a GPS spoofing attack by generating incorrect GPS signals to deceive a GPS receiver~\cite{gps_spoofing_Kugler,geofencing}. Similarly, spoofed radio signals can disrupt the localization process in a WSN. 

In this paper, we present four RSSI-based secure localization techniques. We consider the scenario where a target node attempts to localize itself using RSSI values from its neighboring anchor nodes. It is assumed that the anchors transmit at a fixed predefined power level. However, some of the anchor nodes individually or collaboratively change their transmit power without informing the target node in order to disrupt the localization process.

%\textcolor{blue}{In a collaborative attack the malicious anchor nodes tries to shift the location estimate of the target node to some desired place. These type of attacks are more difficult to handle than the individual attacks. None of the existing techniques can identify the malicious anchor nodes in collaborative and individual attack scenarios. Also, the existing techniques can not find the desired location of the target node in a collaborative attack scenario. It is a vital piece of information as in certain scenarios the desired location may lead to the attackers or will help to understand the motive behind the attack. }

We propose four secure localization techniques: weighted least squares (WLS)~\cite{mukhopadhyay2018robust}, secure weighted least squares (SWLS), and $\ell1$-norm  based localization techniques LN-1 and LN-1E for node location estimation in the presence of malicious anchor nodes using RSSI measurements. We consider both uncoordinated and coordinated attacks by malicious node(s) to disrupt the localization process. We present extensive performance evaluation of the proposed techniques and comparison with two existing secure localization methods, Grad-Desc~\cite{secure_loc_garg} and LMdS~\cite{median_sq_error_li}. We also derive the Cramer-Rao lower bound (CRLB) on the root mean square error (RMSE) of the position estimate under uncoordinated and coordinated attacks.
%Unlike others LN-1E has the ability to estimate the position of the desired location of the target node in a coordinated attack.
%In this work, we propose a weighted least squares (WLS) model for node location estimation in the presence of malicious anchors using received signal strength indicator (RSSI) measurements. We consider uncoordinated attacks where the malicious nodes are assumed to act independently. Extensive performance evaluation of the proposed model and comparison with two existing secure localization methods, Grad-Desc~\cite{secure_loc_garg} and LMdS~\cite{median_sq_error_li}, are presented. The Cramer-Rao lower bound (CRLB) on the root mean square error (RMSE) of the position estimate for the uncoordinated attack scenario is also derived.

The rest of the paper is organized as follows. Section~\ref{Sec:related_work} discusses prior work in the area of secure localization techniques in WSN. Section~\ref{Sec:plb_formulation} presents the problem formulation and Section~\ref{Sec:pro_sec_loc} describes the proposed methods for secure localization. The CRLB for uncoordinated and coordinated attack is derived in Section~\ref{Sec:crlb}. Section~\ref{Sec:perf_eval} presents the performance evaluation of the proposed techniques under uncoordinated and coordinated attack. Section~\ref{Sec:conclusion} concludes the paper. 

\textbf{Notation:} Uppercase bold letters represent matrices and lowercase bold letters represent vectors. 
diag$(\cdot)$ represents a diagonal matrix. $\lVert\cdot\rVert_1$ and $\lVert\cdot\rVert_2$ represent the $\ell_1$-norm and $\ell_2$-norm, respectively. %\textcolor{blue}{$\left|\cdot\right|$ and sign($\cdot$) represent the absolute value and sign of a variable}, and
The cardinality of set $\mathcal{A}$ is denoted by card($\mathcal{A}$).
\section{Related Work}
\label{Sec:related_work}
Lazos and Poovendran~\cite{serloc_lazos} proposed a secure localization technique, named SEcure Range-independent LOCalization scheme (SeRLoc), in which the target node determines its position by using beacon information transmitted by both benevolent and malicious anchor nodes. SeRLoc was shown to be resilient to wormhole attack, Sybil attack, and compromise of network entities. 
%\textcolor{blue}{In this technique, the target node uses the positional coordinate of the anchors and the angles of the antenna boundary lines with respect to a global axis to determine the search area for its location. The target then determines the overlapping area of each antenna sector using a voting scheme and the final location estimate is determined as the centroid of the grid points present in the overlapping region.}
Liu et al.~\cite{voting_secure_liu_2} proposed two range-based localization techniques namely attack-resistant minimum mean square estimation (ARMMSE) and a voting-based scheme. In ARMMSE, malicious anchor nodes are identified by examining the inconsistency among location references indicated by mean square error of estimation. In the voting-based scheme, the network area is divided into a grid and a node votes for a grid cell if the distance of the grid cell to an anchor node is approximately equal to the estimated distance between the target and anchor node. The target node location is estimated as the centroid of the cell with the highest number of votes.  

Li et al.~\cite{median_sq_error_li} considered two robust localization techniques namely, triangulation and RF-based fingerprinting. For triangulation, they proposed an adaptive least squares and least median of squares (LMdS) based location estimation, and for RF fingerprinting they used a median-based distance metric. In LMdS method, anchors are divided into many subsets of identical sizes with each subset estimating the target node location using least squares. The final target node location is given by the least squares location estimate of the subset with the smallest median residue. It was observed that this subset in an attack scenario is least likely to contain malicious nodes. However, it is assumed that the number of malicious nodes is less than $50\%$ of the number of anchor nodes. 
%This technique is computationally expensive as it computes the target node location estimate for each subset. They have also achieved robustness in fingerprint based localization method by using median-based distance metric.    

Garg et al.~\cite{secure_loc_garg} proposed an iterative gradient descent technique with inconsistent measurement pruning to achieve accurate localization in the presence of malicious nodes in a WSN. They consider mobile sensor networks where the nodes are mobile and some of them may be compromised and thus transmit false information. Assuming the measurement noise to be Gaussian the likelihood of the measurements is maximized. 
%The cost function is shown to involve estimated location of the target node (from the previous iteration), and the measured distances between the target node and the anchors. 
To account for the possibility of malicious nodes, the cost function is updated at each iteration by eliminating the anchor nodes with large residues from the localization process.

Jha et al.~\cite{secure_wns_gt_jha} implemented secure localization in WSN using a game theoretic approach. Their proposed technique  combines least trimmed square (LTS) and game theoretic aggregation (GTA) algorithms.  In the LTS phase, the technique learns the weights for each of the anchor nodes. These weights corresponds to the reputation of the anchor nodes i.e., nodes with lower weights are likely to be compromised and malicious. Robust localization is achieved using GTA by filtering out information from the malicious anchor nodes. 
%
%The techniques discussed in ~\cite{voting_secure_liu,voting_secure_liu_2,median_sq_error_li,secure_wns_gt_jha} assumed that the total number of malicious nodes is always less than the total number of anchors. The authors in~\cite{secure_loc_garg} showed results where 60\% of the nodes are malicious and claimed that their technique outperforms~\cite{median_sq_error_li,voting_secure_liu}, although their algorithm always prunes half of anchors to calculate the location of the target node.   
%
\section{System Model}
\label{Sec:plb_formulation}
Consider a network with $N$ anchor nodes whose locations are known, and one or more target nodes whose locations are to be determined. It is assumed that all the anchor nodes are within the communication range of the target node. The nodes are assumed to transmit at a predefined power level. The target node measures RSSI values of the packets received from the anchor nodes and estimates its distance from each of the anchor nodes. Localization techniques allow the target node to estimate its location using the estimated distances and the known locations of the anchor nodes.

Assuming the signal power loss is dominated by path loss which can be modeled using the log-distance model~\cite{Rappaport:2001:WCP:559977}:
%~\cite{Liang_path_loss,Zang_path_loss,Pourhomayoun_path_loss,Empirical_Characterization_Lymberopoulos}:  
%
\begin{IEEEeqnarray}{lCr}\label{Eq:log_normal}
p^r=p_0 - 10n\log_{10}(d) 
\end{IEEEeqnarray}
where $p^r$ is the received power at the target node, $p_0$ is the transmit power of the anchor node, $n$ is the path loss exponent, and $d$ is the distance between the target and the anchor node.
%The value of $n$ is assumed to be known \textit{a priori} and constant throughout the network.}      

Let coordinates of the target and $N$ anchor nodes be represented by $\mathbf{t}=\left[t^x,t^y\right]^T$ and $\mathbf{a}_i=\left[a_i^x,a_i^y\right]^T$  where $i=1,\dots,N$, respectively. The anchor nodes are assumed to broadcast packets at regular intervals. Let $p^r_{ij}$ represent RSSI value of the $j^{\text{th}}$ packet received from the $i^{\text{th}}$ anchor node.
%Let the signal strength of $P$ packets received by the target node from the $i^{\text{th}}$ anchor be represented by $\mathbf{p^{r}_i} = \left[p^{r}_{i_1}, p^{r}_{i_2},...,p^{r}_{i_P}\right]^T$. 
Using (\ref{Eq:log_normal}) and the RSSI value from the $j^{\text{th}}$ packet, the target node computes its distance from the $i^{\text{th}}$ anchor node as $d_{ij}=10^\frac{\left(p_0-p_{ij}^{r}\right)}{10 n}$.
%$\mathbf{d_i}=\left[d_{i_1}, d_{i_2},...,d_{i_P}\right]^T$ where $d_{i_j}=10^\frac{\left(p_0-p_{i_j}^{r}\right)}{10 n}$ and $p^r_{i_j}$ represents RSSI value of the $j^{\text{th}}$ packet received from the $i^{\text{th}}$ anchor node. 
Next consider that some of the anchor nodes are malicious and attempt to disrupt the localization process. Two types of localization attacks are considered:~uncoordinated and coordinated attacks. In an uncoordinated attack, the malicious node(s) act independently and attempt to disrupt the localization process of the target node, whereas in a coordinated attack the malicious nodes coordinate \textit{among themselves} in order to make the target node appear to be located at a location different from its actual location. These attacks are illustrated in Fig.~\ref{fig:plb_f_attack} and can be modeled as described next~\cite{secure_loc_garg,secure_wns_gt_jha}.
\begin{figure}[]
\centering
\begin{subfigure}[b]{0.4\textwidth}
\includegraphics[width=1\linewidth,trim={12cm 0cm 9cm 0},clip=true]{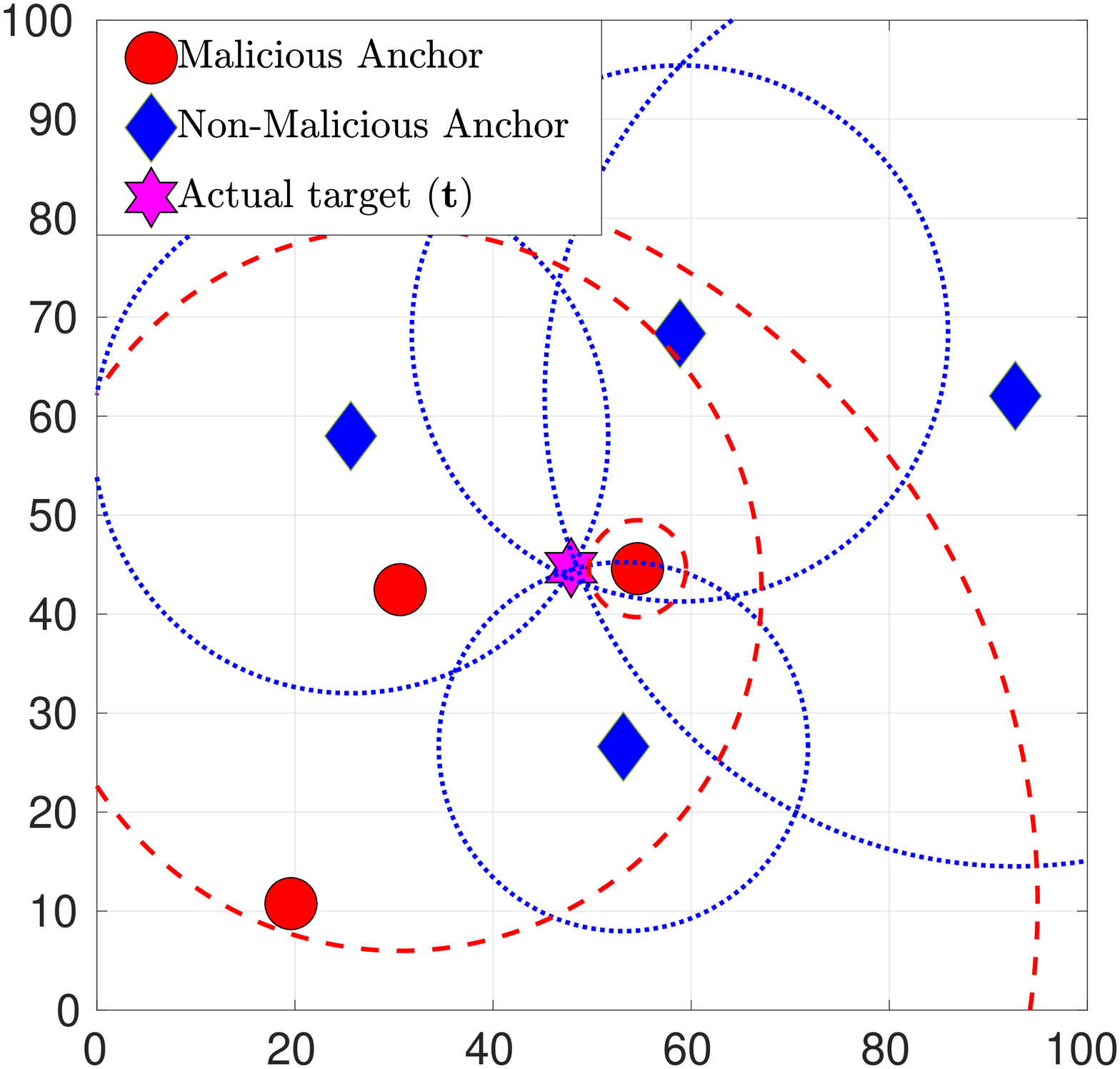}
\caption{Uncoordinated attack}
\label{fig:plb_nc_att}
\end{subfigure}
\begin{subfigure}[b]{0.4\textwidth}
\centering
\includegraphics[width=1\linewidth,trim={12cm 0cm 10cm 0},clip=true]{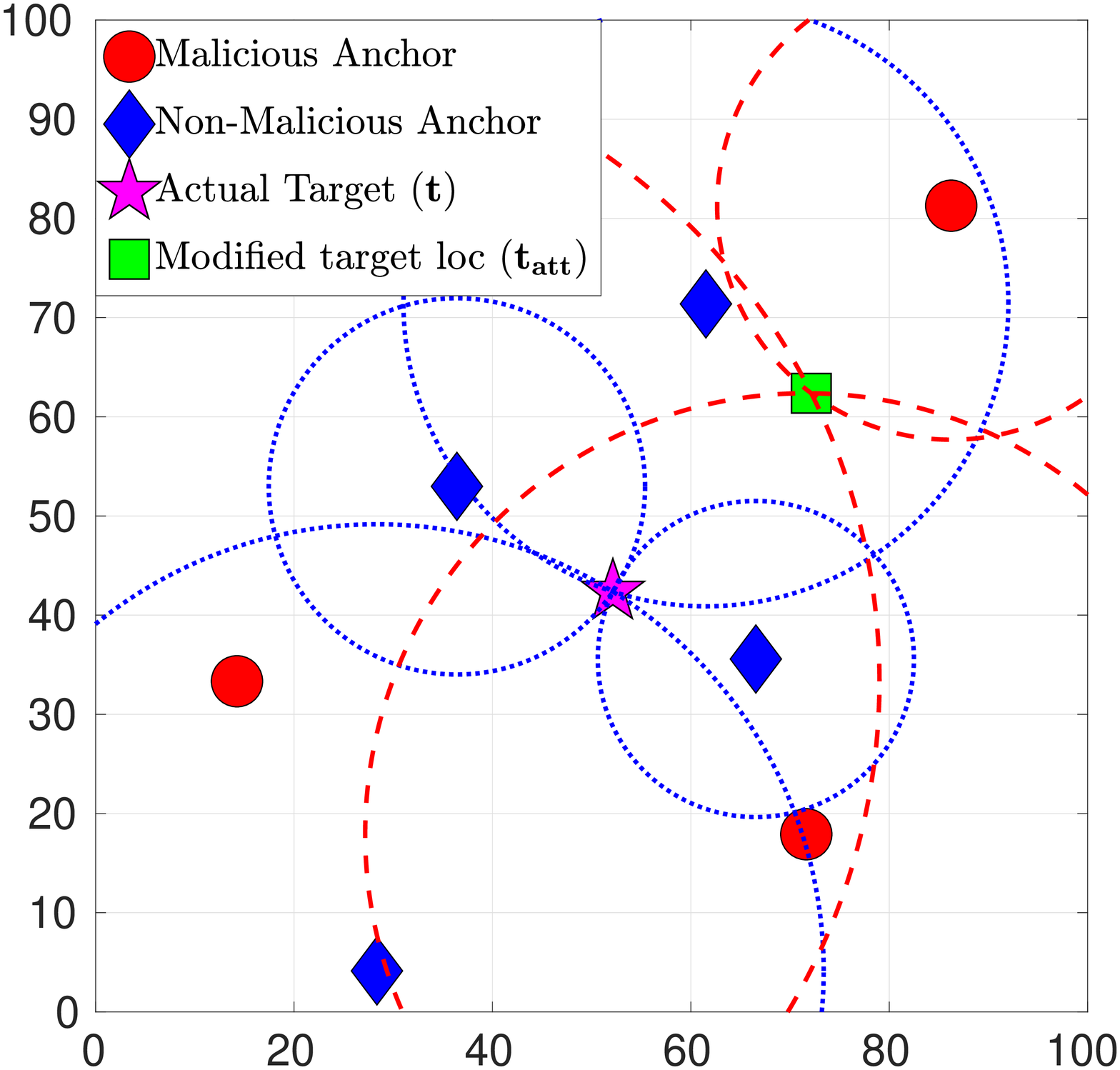}
\caption{Coordinated attack }
\label{fig:plb_c_att}
\end{subfigure}
\caption{Attack strategies by malicious anchor nodes in a WSN. The dashed circles represent the estimated distance between the target and malicious anchor nodes, and the dotted circles represent the estimated distance between the target and non-malicious anchor nodes in the absence of measurement noise.} 
%(assuming measurement noise $\eta = 0$)}
\label{fig:plb_f_attack}
\end{figure}
\subsubsection{Uncoordinated attack}
\label{Sec:non_coordinated_attack}
We consider non-cryptographic attacks where the malicious nodes change their transmit power levels arbitrarily and do not report it to the target node. This type of attack can be modeled as:
\begin{equation}p^r_{i} = \left\{\,\begin{IEEEeqnarraybox}[][c]{l?s}
\IEEEstrut p_0 - 10n\log_{10}(d_i) + \eta & if node $i$ is non-malicious\\
 p_{0_i} - 10n\log_{10}(d_i) + \eta & if node $i$ is malicious,
\IEEEstrut
\end{IEEEeqnarraybox}
\right.\label{eq:model_non_coordinated}
\end{equation}
where $\eta\sim \mathcal{N}(0,\sigma^2)$ is a Gaussian random variable representing measurement noise, $d_i=\left\lVert \mathbf{t} - \mathbf{a}_i\right\rVert_2$ is the distance between the $i^{\text{th}}$ anchor node located at $ \mathbf{a}_i$ and the actual target node located at~$\mathbf{t}$, $p_0$ is the predefined transmit power of the anchor nodes, $p_{0_i}$ is the transmit power of the $i^{\text{th}}$ malicious anchor node, and $p^r_{i}$ is the received power at the target node from the $i^{\text{th}}$ anchor node. Let $p_{0_i}= p_0 + \kappa$ where $\kappa\sim\mathcal{N}(0,\sigma_{\text{att}}^2)$ is a Gaussian random variable with variance $\sigma_{\text{att}}^2$ representing the uncertainty introduced by the malicious anchor nodes in their transmit power levels to disrupt the localization process.

Fig.~\ref{fig:plb_nc_att} illustrates an uncoordinated attack where a target node attempts to localize itself using information from $7$ anchor nodes of which $3$ are malicious. The malicious anchor nodes change their transmit power levels dynamically and since the target node is not aware of their true transmit power $(p_{0_i})$, it incorrectly estimates its distance from the malicious anchor nodes resulting in large error in its estimated location.
%
%The dashed circle around each anchor node in Fig.~\ref{fig:plb_f_attack} represent circle of radii equal to the estimated distance between the target and that anchor node. The circles corresponding to the non-malicious anchor nodes intersect at the true target node location. However, the circles corresponding to the malicious nodes will generally not intersect at the true target node location since the distance estimates will be incorrect. If the target node is not aware of the presence of malicious anchor nodes and uses all the information it receives from both malicious and non-malicious anchor nodes, its position estimate is likely to have large errors.} 
%
\subsubsection{Coordinated attack}
\label{Sec:coordinated_attack}
Coordinated attacks are stronger attacks than the uncoordinated attacks as in coordinated attacks the malicious nodes communicate among themselves with the aim to make the target node appear to be located at a location different from its actual location. It is assumed that the malicious anchor nodes are aware of the actual location of the target node. The coordinated attack can be modeled as:    
\begin{equation}p^r_{i} = \left\{ \,\begin{IEEEeqnarraybox}[][c]{l?s}
\IEEEstrut p_0 - 10n\log_{10}(d_i) + \eta & if node $i$ is non-malicious,  \\
p_{0_i}^{c} - 10n\log_{10}( d_i) + \eta  & if node $i$ is malicious 
\IEEEstrut
\end{IEEEeqnarraybox}
\right.\label{eq:model_coordinated}
\end{equation}
where $d_i = \left\lVert\mathbf{t}-\mathbf{a}_i\right\rVert_2$, $p_{0_i}^{c} = p_0 - 10n\log_{10}(\chi_i)$ with $\chi_i=\frac{\left\lVert\mathbf{t}_{\text{att}}-\mathbf{a}_i\right\rVert_2}{\left\lVert\mathbf{t}-\mathbf{a}_i\right\rVert_2}$. $\mathbf{t}_{\text{att}}$ is the location where the malicious anchor nodes are trying to make the target node appear to be located, and $p_{0_i}^{c}$ is the transmit power of the $i^{\text{th}}$ malicious anchor node. Thus, $\chi_i$ is the factor by which the malicious anchor nodes scale the actual distance between themselves and the target node.     
%$\chi$ is the factor by which the malicious nodes scale the actual distance between themselves and the target node, 

A coordinated attack is shown in Fig.~\ref{fig:plb_c_att} where 3 malicious anchor nodes attempt to make the target node appear to be located at $\mathbf{t}_{\text{att}}$ instead of its actual location $\mathbf{t}$. It is seen that the dashed circles, with radius equal to the erroneous distance between the malicious anchor nodes and the target node, intersect at $\mathbf{t}_{\text{att}}$, whereas the dotted circles with radius equal to the actual distance between the non-malicious anchor nodes and the target node intersect at $\mathbf{t}$.
\section{Proposed techniques for secure localization}
\label{Sec:pro_sec_loc}
The variation in the received power ($p^r$) as a function of the distance ($d$) between the anchor and target node is shown in Fig.~\ref{fig:proposed_model}. From~(\ref{Eq:log_normal}): 
\begin{IEEEeqnarray}{lCr}\label{Eq:distance_est_noise_free}
d=10^{\left(\frac{p_0-p^r}{10 n}\right)} \Rightarrow d=c \ 10^\frac{-p^r}{10 n}
\end{IEEEeqnarray}
where $c=10^{\frac{p_0}{10n}}$. As the relationship between $p^r$ and $d$ is non-linear, variations in the received power affect the estimated distance in a non-linear manner. From Fig.~\ref{fig:proposed_model}, it is seen that the distance estimates are more robust to noise when the anchor nodes are close to the target node than when the anchor nodes are farther from the target node. Thus, anchor nodes closer to the target node should be given a larger weight in the localization process than anchor nodes that are farther.
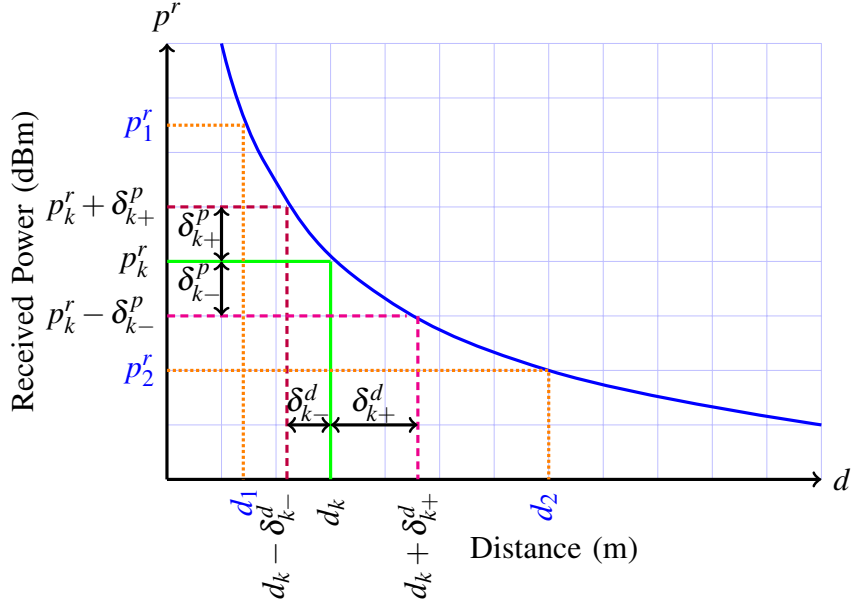
\begin{figure}[]
\centering
\begin{tikzpicture}[scale=1.45]
%\draw [ thick]  (-1.5,4.45) rectangle (6.35,-1.1) node (v17) %{};

\def\xmin{0}
  \def\xmax{6}
  \def\ymin{0}
  \def\ymax{4}

  % grid
  \draw[style=help lines,blue!25, ystep=.5, xstep=.5] (\xmin,\ymin) grid
  (\xmax,\ymax);
 % axes
  \draw[->,very thick,black] (\xmin,\ymin) -- (\xmax,\ymin) node[right] {$d$};
  \draw[->,very thick,black] (\xmin,\ymin) -- (\xmin,\ymax) node[above] {$p^r$};
 
  %\draw[blue] (0.5,0.8) parabola (-5.5,4) node[left,black] {$D$};
%\draw plot [domain=0.1:1,samples=100] (\x,{log10(\x)});

\draw [blue, very thick] plot[smooth, tension=.7,very thick] coordinates {(0.5,4) (0.9,2.9) (1.8,1.8) (3.5,1) (6,0.5)};

\node [black] at (3.55,-0.65) { {Distance (m)}};
\node [black,rotate=-270] at (-1.3,2) { {Received Power (dBm)}};

\node (v1) at (1.5,2.1) {};
\node (v2) at (1.5,-0.1) {};
\draw [green,very thick] (v1) edge (v2);

\node [black,rotate=-270] at (1.5,-0.3) { {$d_k$}};
\node [black,rotate=-270] at (1,-0.6) { {$d_k-\delta^d_{k-}$}};
\node [black,rotate=-270] at (2.3,-0.6) { {$d_k+\delta^d_{k+}$}};

\node (v3) at (2.3,1.6) {};
\node (v4) at (2.3,-0.1) {};
\draw [very thick,densely dashed,magenta] (v3) edge (v4);

\node (v5) at (1.1,2.6) {};
\node (v6) at (1.1,-0.1) {};
\draw [ very thick,densely dashed,purple]  (v5) edge (v6);

\node [black,rotate=0] at (-0.6,1.5) { {$p^r_k-\delta^{p}_{k-}$}};
\node [black,rotate=0] at (-0.3,2) { {$p^r_k$}};
\node [black,rotate=0] at (-0.6,2.5) { {$p^r_k+\delta^{p}_{k+}$}};

\node (v8) at (1.6,2) {};
\node (v7) at (-0.1,2) {};
\draw [very thick,green] (v7) edge (v8);

\node (v10) at (1.2,2.5) {};
\node (v9) at (-0.1,2.5) {};
\draw  [very thick, densely dashed,purple ](v9) edge (v10);

\node (v12) at (2.3,1.5) {};
\node (v11) at (-0.1,1.5) {};
\draw  [very thick, densely dashed,magenta ] (v11) edge (v12);
\node (v13) at (1.4,0.5) {};
\node (v14) at (2.4,0.5) {};
\node (v16) at (0.5,1.9) {};
\node (v15) at (0.5,2.6) {};

\draw [very thick,<->] (v13) edge (v14);
\draw [very thick,<->]  (v15) edge (v16);
\node [black,rotate=0] at (0.3,1.8) { {$\delta^p_{k-}$}};
\node [black,rotate=0] at (0.3,2.3) { {$\delta^p_{k+}$}};
\node [black,rotate=0] at (1.3,0.7) { {$\delta^d_{k-}$}};
\node [black,rotate=0] at (1.9,0.7) { {$\delta^d_{k+}$}};

\node [blue,rotate=0] at (-0.25,3.25) { {$p^r_{1}$}};
\node [blue,rotate=0] at (-0.25,1) { {$p^r_{2}$}};

\node [blue,rotate=-270] at (0.7,-0.2) { {$d_{1}$}};
\node [blue,rotate=-270] at (3.45,-0.25) { {$d_{2}$}};
\draw [densely dotted,very thick,orange](0,3.25) -- (0.7,3.25) -- (0.7,0);
\draw [densely dotted,very thick,orange](0,1) -- (3.5,1) -- (3.5,0);
\node (v17) at (1,0.5) {};
\node (v18) at (1.6,0.5) {};
\draw [<->,very thick,black] (v17) edge (v18);
\node (v19) at (0.5,2.1) {};
\node (v20) at (0.5,1.4) {};
\draw [<->,very thick,black] (v19) edge (v20);
\end{tikzpicture} 
\caption{Effect on distance estimation due to variations in the received power.}
\label{fig:proposed_model}
\end{figure}

In Fig.~\ref{fig:proposed_model}, a target node that receives packet with power $p^r_k$ from the $k^{\text{th}}$ anchor node is estimated to be at a distance $d_k$ from it. Positive and negative variations in the received power $\left(p^r_k\right)$ are represented by $\delta^p_{k+}$ and $\delta^p_{k-}$, and the corresponding variation in distance estimates are represented by $\delta^d_{k-}$ and $\delta^d_{k+}$. Lemma~\ref{lemma:dist_per} quantifies the variation in the estimated distance corresponding to a variation in the received power. 
\begin{lem}
\label{lemma:dist_per}
%Positive and negative variation in the received power $p^r_k$ is expressed as $\delta^p_{k+}$ and $\delta^p_{k-}$. $\delta^d_{k+}$ and $\delta^d_{k-}$ represents the positive and negative variation in distance from $d_k$.
%
%\forall \  \delta^{p^r}_k > 0, \  \forall \  k \in \mathbb{N^+}   \Rightarrow \exists \ \delta^{d}_k > 0, \ 
$\delta^d_{k-}= g\left(\delta^p_{k+}\right)10^{-\frac{p^r_k}{10n}}$ and $\delta^d_{k+}= -g\left(-\delta^p_{k-}\right)10^{-\frac{p^r_k}{10n}}$ where $g\left(x\right) = c\left(1 - 10^\frac{-x}{10n}\right)$ and $\delta^d_{k+},\delta^d_{k-},\delta^p_{k+},\delta^p_{k-}$ are all positive quantities.
%$g\left(\delta^p_k\right) = c\left(1 - %10^\frac{-\delta^p_k}{10n}\right)$
\end{lem}
\begin{proof}
%\begin{IEEEeqnarray}{lCr}\label{Eq:dist_pertubated_power}
%\nonumber
$d_k-\delta^d_{k-} = c 10^{-\frac{p^r_k+\delta^p_{k+}}{10n}},
\delta^d_{k-} = c  10^\frac{-p^r_k}{10 n} - c 10^{-\frac{p^r_k+\delta^p_{k+}}{10n}}  = 10^{-\frac{p^r_k}{10n}} c \left(1 - 10^\frac{-\delta^p_{k+}}{10n}\right) = 10^{-\frac{p^r_k}{10n}}g\left(\delta^p_{k+}\right)$. Similarly, $d_k+\delta^d_{k+} = c 10^{-\frac{p^r_k-\delta^p_{k-}}{10n}} \Rightarrow \delta^d_{k+} = - 10^{-\frac{p^r_k}{10n}}g\left(-\delta^p_{k-}\right)$.
%\end{IEEEeqnarray}
\end{proof}
\begin{lem}
\label{lemma:same_power_two_diff_dist}
For the same amount of variation ($\delta_k^p$) on the received power $p^r_k$, the received power with negative deviation ($p^r_k - \delta_k^p$) will result in larger variation in the distance estimate than the received power with positive deviation ($p^r_k + \delta_k^p$) i.e., if $\delta^p_{k+} = \delta^p_{k-}=\delta^p_k$, then  $\delta^d_{k+} > \delta^d_{k-}$.
\end{lem}
\begin{proof}
Using Lemma~\ref{lemma:dist_per} it can be shown that
%$g(x)+g(-x)$ is always negative. The entity $g(x)+g(-x)$ can be expressed as
%
\begin{IEEEeqnarray}{lcr}
\label{Eq:g_x_plus_g_x_is_neg}
g(x)+g(-x)=c\left( 2 - 10^{\frac{-x}{10n}}- 10^{\frac{x}{10n}}\right) = -\underbrace{\left(c10^{\frac{-x}{10n}} \right)}_{a} \underbrace{\left( 10^{\frac{x}{10n}} - 1 \right)^2}_{b} \leq 0
\end{IEEEeqnarray}
$g(x)+g(-x)$ is always negative because $a$ and $b$ are always positive for all $x \in \mathbb{R}$. 

The variation in distance corresponding to the negative and positive change in received power $p^r_k$ are given by $d_k+\delta_{k+}^d$ and $d_k-\delta_{k-}^d$, respectively. From Lemma~\ref{lemma:dist_per}:
\begin{IEEEeqnarray}{lCr}
\label{Eq:deviation_distance}
\delta_{k+}^d = -g(-\delta^p_k)10^{\frac{p^r_k}{10n}} \text{ and } \delta_{k-}^d = g(\delta^p_k)10^{\frac{p^r_k}{10n}}
\end{IEEEeqnarray}
Using (\ref{Eq:g_x_plus_g_x_is_neg}) and (\ref{Eq:deviation_distance}), it is seen that  $\delta_{k+}^d \geq \delta_{k-}^d$ when $\delta^p_{k+} = \delta^p_{k-}=\delta^p_k$ (refer Fig.~\ref{fig:proposed_model}).
\end{proof}
Lemma~\ref{lemma:two_diff_dist_per} states that the distance estimation using RSSI values from anchor nodes that are farther from the target node is less robust to variation in the received power.
\begin{lem}
\label{lemma:two_diff_dist_per}
For the same amount of variation at two different received power levels, the lower received power level will result in larger variation in the distance estimate than the higher received power level i.e., if $p^r_{1} > p^r_{2}$ and $\delta_1^p = \delta^p_2, \ \text{then} \ \delta^d_{2} > \delta^d_{1}$.
\end{lem}
\begin{proof}
If variation in the received power levels is positive i.e., $p^r_1 + \delta_1^p$ and $p^r_2 + \delta_2^p$ where $\delta_1^p = \delta^p_2$, then the relationship between the corresponding variations in distance estimates is $\delta^d_{2} > \delta^d_{1}$. Using Lemma~\ref{lemma:dist_per}:
\begin{IEEEeqnarray}{lCr}
p^r_{1} > p^r_{2} \Rightarrow %10^{\frac{-p^r_{2}}{10n}} > %10^{\frac{-p^r_{1}}{10n}} \Rightarrow 
g\left(\delta^p_2\right) 10^{\frac{-p^r_{2}}{10n}} > g\left(\delta^p_1\right) 10^{\frac{-p^r_{1}}{10n}} 
 \Rightarrow \delta^d_{2} > \delta^d_{1}
\end{IEEEeqnarray}
This relationship also holds if the variation in the received power levels is negative.
%then $p^r_{1} > p^r_{2} \Rightarrow -g\left(-\delta^p_2\right) 10^{\frac{-p^r_{2}}{10n}} > -g\left(-\delta^p_1\right) 10^{\frac{-p^r_{1}}{10n}} \Rightarrow \delta^d_{2} > \delta^d_{1}$.
%
\end{proof}

From Lemma~\ref{lemma:two_diff_dist_per}, we can say that the malicious nodes that are closer to the target node can have a greater impact in disrupting the localization process than nodes that are farther. To reduce the effect of the malicious anchor nodes we next propose secure localization techniques.
\subsection{Weighted Least Square (WLS)}
We propose a secure localization technique based on the weighted least squares algorithm~\cite{mukhopadhyay2018robust}. This is a modified version of the least squares~\cite{kay1993fundamentals} where the anchor nodes are assigned weights based on their distance from the target node. The target receives $P$ packets from each of the anchor nodes and computes the mean received power as $\overline{p_i^r}$~$=\frac{1}{P}\sum_{j=1}^{P} p^r_{ij}$.% where $p_{ij}^r$ represents RSSI value of $j^{th}$ packet received from the $i^{th}$ anchor node. 

Using $\overline{p_i^r}$, the target node estimates its distance from the $i^{\text{th}}$ anchor node as $\overline{d_i}$. Given $\mathbf{a}_i$ and $\overline{d_i}$ for $i=1,\dots,N$, the target node position estimation can be formulated as a weighted least squares problem where each of the distance estimates is weighted by the variance of $\overline{d_i}^2$. Assuming the distance estimate to be a random variable, its cumulative distribution function (CDF) is: 
\begin{IEEEeqnarray}{rCl}
\label{Eq:CDF_distance}
 P (d_i\leq\gamma) = P\left(\frac{\eta}{\sigma}\leq\frac{p_i^r - p_0 + 10n\log_{10}\gamma}{\sigma}\right) = 1-Q\left(f(\gamma)\right)\nonumber
\end{IEEEeqnarray}
where $Q(\cdot)$ is the Q-function and $f(\gamma) = \frac{10n}{\sigma}\log_{10}\left(\frac{\gamma}{d_i}\right)$. The probability density function (PDF) of $d_i$ can be shown to be:
\begin{IEEEeqnarray}{l}
\label{Eq:pdf}
    f_{\text{d}_i}(\gamma)= \frac{5 n}{\gamma\sigma\ln(10)}\sqrt{\frac{2}{\pi}}\exp\left(-\frac{50 n^2\ln^2\left(\frac{\gamma}{d_i}\right)}{\sigma^2\ln^{2}(10)}\right)
\end{IEEEeqnarray}
Using~\eqref{Eq:pdf},  the variance of $d_i$ and $d_i^2$ can be shown to be: 
\begin{IEEEeqnarray}{l}
\label{Eq:var_d}
f^{\text{Var}}_{\text{d}_i}(d_i,\sigma)= d_i^2\exp\left(\frac{\sigma^2 }{18.86n^2}\right)\left[\exp\left(\frac{\sigma^2 }{18.86n^2}\right)-1\right]\\
\label{Eq:var_d^2}
f^{\text{Var}}_{\text{d}_i^2}(d_i,\sigma)= d_i^4\exp\left(\frac{\sigma^2}{4.715n^2}\right)\left[\exp\left(\frac{\sigma^2 }{4.715n^2}\right) -1 \right]
\end{IEEEeqnarray}
 
The measurement matrix $\mathbf{P}^\text{r} \left( =\left[p^{r}_{ij}\right]  \text{ where } i=1,\dots,N \text{ and } j=1,\dots,P \right)  $ consists of the RSSI values of the $P$ packets received by the target node from each of the $N$ anchor nodes (both malicious and non-malicious nodes). Consider a diagonal weighing matrix $W$ whose elements are inverse of the variance of $d_i^2$. The measurement model can be expressed as $\mathbf{At}+\mathbf{\epsilon}=\mathbf{b}$ where $\mathbf{A}$ and $\mathbf{b}$ are given in terms of the anchor node positions ($\mathbf{a}_i$) and distance estimates ($\overline{d_i}$) as:
\begin{IEEEeqnarray}{l}
\label{Eq:A_B}
\mathbf{A}=
\begin{pmatrix}
-2a_1^x & -2a_1^y & 1 \\
-2a_2^x & -2a_2^y &  1 \\
\vdots & \vdots & \vdots \\
-2a_N^x & -2a_N^y & 1
\end{pmatrix},\,
\mathbf{b}=
\begin{pmatrix}
\overline{d_1}^2 - (a_1^x)^2 - (a_1^y)^2 \\
\overline{d_2}^2 - (a_2^x)^2 - (a_2^y)^2 \\
\vdots \\
\overline{d_N}^2 - (a_N^x)^2 - (a_N^y)^2 \\
\end{pmatrix}
\end{IEEEeqnarray} 
Assuming the measurement noise $\mathbf{\epsilon}$ to be Gaussian distributed, estimation of the target node position can be formulated as a weighted least squares problem and Algorithm~\ref{Algo:WLMS} presents the WLS secure localization technique.
\begin{algorithm}
\caption{WLS Localization}
\label{Algo:WLMS}
\begin{algorithmic}[1]
\renewcommand{\algorithmicrequire}{\textbf{Input:}}
\renewcommand{\algorithmicensure}{\textbf{Output:}}
\REQUIRE $\mathbf{P}^\text{r}$, $\mathbf{a}_i$ ($i=1, \dots, N$), $\sigma$
\ENSURE $\hat{\mathbf{t}}=\left[t^x,t^y\right]^T$\\
\textit{Initialize}: $\mathbf{A},\mathbf{b}$
\STATE $\overline{\mathbf{p}^{\text{r}}}=\frac{1}{P}\left[\sum_{j=1}^P p^r_{1j}, \ \sum_{j=1}^P p^r_{2j},\dots,\sum_{j=1}^P p^r_{Nj}\right]^T$
\STATE $\overline{\mathbf{d}}=\left[\overline{d_1},\overline{d_2},\dots,\overline{d_N}\right]^T = 10^{\frac{\left({p_{0}}-\overline{\mathbf{p}^{\text{r}}}\right)}{10n}}$
%\FOR {$i = 1$ to $N$}
\STATE $\sigma_{\text{d}^2}^2(i) = f^{\text{Var}}_{\text{d}_i^2} \left(\overline{d_i},\sigma\right), i=1,..,N$ (using~\eqref{Eq:var_d^2}) %\COMMENT{using~\eqref{Eq:var_d^2}}
%\ENDFOR
\STATE $\mathbf{W}$ = diag$ \left (1/{\sigma_{\text{d}^2}^2}(1),1/{\sigma_{\text{d}^2}^2}(2),\dots,1/{\sigma_{\text{d}^2}^2}(N)\right)$
\STATE $\mathbf{\hat{q}}= \left({\mathbf{A}}^T\mathbf{W}{\mathbf{A}}\right)^{-1} \mathbf{A}^T \mathbf{W} \mathbf{b}$
 \STATE $\hat{\mathbf{t}} = \left[\hat{q}_{1},\hat{q}_{2}\right]^{T}$
\end{algorithmic} 
\end{algorithm}
\subsection{Secure Weighted Least Square (SWLS)}
Consider an uncoordinated attack scenario where the malicious nodes change their transmit powers arbitrarily without communicating this to the target node (refer Section~\ref{Sec:non_coordinated_attack}). The received power $\left(p_{0_i}\right)$ from a malicious anchor node can be expressed using (\ref{eq:model_non_coordinated}) as: 
\begin{IEEEeqnarray}{l}
\label{Eq:modified_non_co_attack}
p_i^r=p_{0} - 10n\log_{10}(d_i) + \eta + \kappa 
\end{IEEEeqnarray}
Thus, the RSSI values from the malicious and non-malicious anchor nodes have standard deviation of $\sqrt{\sigma^2+\sigma_{\text{att}}^2}$ and $\sigma$, respectively. SWLS attempts to identify the malicious anchor nodes by observing the RSSI values, and eliminates them from the localization process. The detailed procedure of the SWLS localization is given in Algorithm~\ref{Algo:SWLMS}. 
\begin{algorithm}
\caption{SWLS Localization}
\label{Algo:SWLMS}
\begin{algorithmic}[1]
\renewcommand{\algorithmicrequire}{\textbf{Input:}}
\renewcommand{\algorithmicensure}{\textbf{Output:}}
\REQUIRE $\mathbf{P}^\text{r}$, $\mathbf{a}_i$ ($i=1, \dots, N$), $\sigma$, $\zeta$ %w, $\sigma_{\text{est}}^{\text{max}}$ %$\sigma_{\text{max}}$,
\ENSURE  $\hat {\mathbf{t}}=[t^x,t^y]^T$\\ 
\textit{Initialize}: $\boldsymbol{\mathbf{A}}$, $\mathbf{b}$, $\mathcal{M} = \emptyset$ 
% \\ \underline{Finding Non-Malicious Anchors:}
\STATE $\mathbf{D}=10^{\frac{P_{0}-\mathbf{P}^\text{r}}{10n}}$
\STATE $\overline{\mathbf{p^r}}=\frac{1}{P} \left[\sum_{j=1}^P p^r_{1j},\ \sum_{j=1}^P p^r_{2j},\ \dots \, \ \sum_{j=1}^P p^r_{Nj}\right]^T$
\STATE $\overline{\mathbf{d}}=10^{\frac{P_{0}-\mathbf{\overline{p^r}}}{10n}} = \left[\overline{d_1}, \ \overline{d_2}, \dots, \overline{d_N} \right]^T$
\FOR {$i = 1$ to $N$}
\STATE  $\hat{\sigma}_{\text{est}}=\underset{\sigma_{est} \geq 0}{\arg\min} \ \left|\text{Var}(d_{i1}, d_{i2}, \dots, d_{iP}) - f^{\text{Var}}_{\text{d}} (\overline{d_i},\sigma_{\text{est}}) \right|$, using (\ref{Eq:close_form_sigma_est})
%$\left(\sigma_{est}=0:\text{w}:\sigma_{est}^{\text{max}}\right)$
%\IF {($\hat{\sigma}_{est}$ < $\zeta\sigma$)}
%\STATE $\textit{M}=\textit{M} \cup \{i\}$  %\COMMENT{index of non-malicious nodes}
%\ENDIF
\STATE $\mathcal{M}= \mathcal{M} \cup \{ i \ | \ \hat{\sigma}_{est} < \zeta\sigma \}$ (index of non-malicious nodes)
\ENDFOR
% \\ \underline{Localization:}
%\FOR {$k \text{ in } \textit{M}$}
\STATE $\mathbf{\sigma}_{{d}^2}^2(k)$ = $f^{\text{Var}}_{d^2}\left(\mathbf{\overline{d}}_k,\sigma \right) \text{ where }k \in \mathcal{M} $, using (\ref{Eq:var_d^2})  
%\COMMENT{using (\ref{Eq:var_d^2})}
%\ENDFOR
\STATE $\mathbf{\hat{A}} = \mathbf{A}(j,:)$, $\mathbf{\hat{b}} = \mathbf{b}(j,:)$, $j \in \mathcal{M}$
\STATE $\mathbf{{W}}$ = diag$\left(1/\mathbf{\sigma}_{{d}^2}^2(1),\ 1/\mathbf{\sigma}_{{d}^2}^2(2),\dots,  1/\mathbf{\sigma}_{{d}^2}^2(\text{card}(\mathcal{M}))\right)$
\STATE $\hat{\mathbf{q}}= \left({\mathbf{\hat{A}}}^T\mathbf{ {W}}{\mathbf{\hat{A}}}\right)^{-1} \mathbf{\hat{A}}^T \mathbf{ {W}} \mathbf{\hat{b}}$
\STATE $\hat{\mathbf{t}} = \left[\hat{q}_{1},\hat{q}_{2}\right]^{T}$
\end{algorithmic} 
\end{algorithm}
Let $\mathbf{D}$ be the $P \times N$ matrix of distance estimates from the target node to each of the anchor nodes i.e., $\mathbf{D}=\left [d_{ij}\right]$ where $i=1,\dots,N$ and $j=1,\dots,P$, computed using (\ref{Eq:distance_est_noise_free}). The average distance $\left(\overline{d_i}\right)$ between the target node and an anchor node is calculated using $\overline{p^r}$ obtained from the RSSI values of $P$ packets and not as a column average of the matrix $\mathbf{D}$. This is due to the fact that the estimated distances $\left(d_i\right)$ do not follow a Gaussian distribution (refer~(\ref{Eq:pdf})).
%
%Distances of the target node from each of the anchors are estimated with the help of $\mathbf{P}^\text{r}$ (using Eq.\ref{Eq:distance_est_noise_free}) and are stored in matrix $\mathbf{D}$. 
%
% \begin{IEEEeqnarray}{l}
%   \nonumber
%   \mathbf{D} =
%   \begin{pmatrix}
%  d_{1_1} & d_{2_1} & \hdots & d_{N_1} \\
%  d_{1_2} & d_{2_2} & \hdots & d_{N_2} \\
%  \vdots & \vdots & \vdots \\
%  d_{1_P} & d_{2_P} & \hdots & d_{N_P}
%   \end{pmatrix}
%  \end{IEEEeqnarray}  
% $\mathbf{D}$ is used to calculate the variation of distances encountered by target node from each of the anchors.
%obtained from $P$ packets is calculated using $\overline{p^r}$ and Eq.~\ref{Eq:distance_est_noise_free}, 
%whereas received power ($p^r$) has a normal distribution $\mathcal{N}(p_0-10n\log(d),\sigma^2)$ for non-malicious anchors and $\mathcal{N}(p_0-10n\log(d),\sigma^2+\sigma_{test}^2)$ for malicious anchors.    

Variance of the distance estimates is calculated based on the estimated distances from each of the $P$ packets and using (\ref{Eq:var_d}) for different values of the noise standard deviation. An estimate of the noise standard deviation $\left(\sigma_{\text{est}}\right)$ is obtained for each of the anchor nodes by minimizing $\left|\text{Var}(d_{i1}, d_{i2},\dots,d_{iP}) - f^{\text{var}}_{d} (\overline{d_i},\sigma_{\text{est}})\right|$ where $\text{Var}(d_{i1}, d_{i2},\dots,d_{iP})$ represents the variance of the distance estimates calculated from each of the $P$ packets. The malicious nodes are identified by applying a threshold on $\hat{\sigma}_{\text{est}}$ with the threshold level set to $\zeta\sigma$ where $\zeta > 0$. 
Finally, weighted least square localization (as discussed in Algorithm~\ref{Algo:WLMS}) is applied by considering only the non-malicious anchor nodes. The closed-form solution of the estimate $\hat{\sigma}_{\text{est}}$ is given by (used in line 5 of Algorithm~\ref{Algo:SWLMS}):
\begin{IEEEeqnarray}{lCR}
\label{Eq:close_form_sigma_est}
% \hat{\sigma}_{\text{est}}=\sqrt{ 
% 18.86n^2 \ln \left(  
% \frac{1+
% \sqrt{1+ \left( 4\text{Var}(d_{i1}, d_{i2}, \dots, d_{iP})/\overline{d_i}^2\right)}}{2}\right)}\\
\hat{\sigma}_{\text{est}}=\sqrt{18.86n^2\ln\left(0.5+ 0.5\sqrt{1+ \frac{4\text{Var}(d_{i1},d_{i2},\dots, d_{iP})}{\overline{d_i}^2}}\right)}
\end{IEEEeqnarray}
The derivation of (\ref{Eq:close_form_sigma_est}) is provided in the Appendix~\ref{appendix}.

This technique relies on the variance in the RSSI values, and in the case of a coordinated attack the variance remains the same for malicious and non-malicious nodes (refer~(\ref{eq:model_coordinated})). Thus, this technique is not robust to coordinated attacks on the localization process. 
\subsection{Localization using $\ell_1$-norm Optimization }
The localization problem can be posed as a 3-dimensional plane fitting problem $z=f(x,y)$ where $z$ represents $\mathbf{b}$, and $x$, $y$ represent the first two columns of $\mathbf{A}$ (refer~(\ref{Eq:A_B})). The objective is to find a plane $z=\alpha x +\beta y + \gamma$ where $\alpha = t^x$, $\beta = t^y$, and $\gamma = (t^x)^2+(t^y)^2$, that fits the measurements (or data points) $\langle -2a_i^x,-2a_i^y,  \overline{d_i}^2 - (a_i^x)^2 - (a_i^y)^2\rangle$, $i=1,\dots,N$. The values of $\alpha$, $\beta$, and $\gamma$ can be obtained by minimizing the $\ell_2$-norm based distance between the measurements and the plane: %(refer~(\ref{Eq:LN_norm2})). 
\begin{IEEEeqnarray}{lCr}
\label{Eq:LN_norm2}
\begin{aligned}
& \underset{\mathbf{u}}{\text{min}} & & \left\lVert\mathbf{r}\right\rVert^2_2 \\
& \text{subject to} & & \mathbf{r}=\mathbf{Au} - \mathbf{b} \\
\end{aligned}
\end{IEEEeqnarray}
where $\mathbf{u}=[\alpha, \beta, \gamma ]^T$. The closed form solution of (\ref{Eq:LN_norm2}) is $\mathbf{u}=\mathbf{A}^\dagger \mathbf{b}$~\cite{boyd2004convex} where $\dagger$ represents the pseudo inverse. This is similar to the localization process discussed in Algorithm~\ref{Algo:WLMS}.   

In an uncoordinated attack scenario, the measurements representing points in three dimensions are divided into two categories on the basis of the variance in the received power. Measurements from non-malicious nodes display less variance than those from malicious nodes. From a curve fitting perspective, the data points corresponding to the malicious nodes can be treated as outliers.
%with respect to the data points of non-malicious nodes. 
%
\begin{figure}
%\begin{subfigure}[]{\textwidth}
\centering
\includegraphics[width=12cm,keepaspectratio]{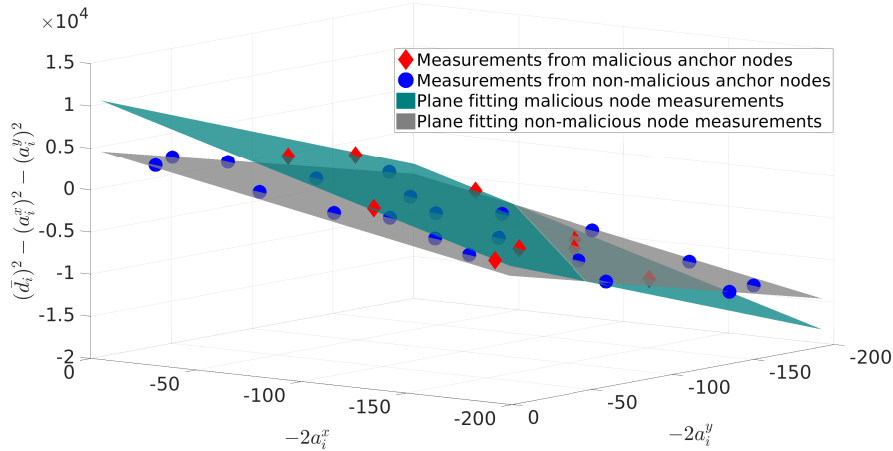}
%\caption{Planes containing measurements corresponding to the malicious and non-malicious anchor nodes. \big($ \left\lVert\mathbf{t}-\mathbf{t_{\text{att}}}\right\rVert_2=35.35$~m\big) }
%\label{fig:LN_1_two_plane}
%\end{subfigure}
%
% \begin{subfigure}[]{\textwidth}
% \centering
% \includegraphics[width=12cm,keepaspectratio]{pictures/LN_1_plane_dist_error__7_07_28_28.eps}
% \caption{Measurements reported by the malicious and non-malicious anchor nodes. In Scenario 1 $\left\lVert\mathbf{t}-\mathbf{t}_{\text{att}}\right\rVert_2=7.07$~m, and in Scenario 2 $\left\lVert\mathbf{t}-\mathbf{t}_{\text{att}}\right\rVert_2=28.28$~m.}
% \label{fig:LN_1_non_co_dist_error_7_28}
% \end{subfigure}
\caption{Planes containing measurements corresponding to the malicious and non-malicious anchor nodes. Anchor nodes are randomly distributed with 31\% of the anchor nodes being malicious, measurement noise ($\sigma$) is 0~dB, and $ \left\lVert\mathbf{t}-\mathbf{t}_{\text{att}}\right\rVert_2=35.35$~m.}
\label{fig:LN_1_non_co_plane}
\end{figure}
However, in a coordinated attack, measurements from all the nodes display similar variance since the malicious nodes do not vary their transmit power randomly (refer~(\ref{eq:model_coordinated})). The transmit power of the malicious anchor nodes is assumed to be different from those of the non-malicious anchor nodes, and depends on the value of $\chi_i$. Fig.~\ref{fig:LN_1_non_co_plane} shows a coordinated attack visualized as a plane fitting problem. The data points corresponding to the measurements from malicious and non-malicious anchor nodes are shown to lie on two different planes. Determining these planes will enable us to estimate the location of the target node ($\mathbf{t}$) and the location where the malicious anchor nodes intend to make the target node appear to be located ($\mathbf{t}_{\text{att}}$). %Fig.~\ref{fig:LN_1_non_co_dist_error_7_28}  shows the measurement points for two different scenarios, with $\left\lVert\mathbf{t}-\mathbf{t}_{\text{att}}\right\rVert_2$ being 7.07~m and 28.28~m. 
As $\left\lVert\mathbf{t}-\mathbf{t}_{\text{att}}\right\rVert_2$ increases the measurements corresponding to the malicious anchor nodes move farther from the plane that fits the measurements from non-malicious anchor nodes. 

In a coordinated attack scenario, measurements from the malicious anchor nodes can thus be treated as outliers. Hence, under both uncoordinated and coordinated attacks, the optimization problem (\ref{Eq:LN_norm2}) can be rewritten as a robust plane fitting problem by replacing the $\ell_2$-norm in the objective with an $\ell_1$-norm~\cite{boyd2004convex,baraniuk2011introduction}:
\begin{IEEEeqnarray}{l}
\label{Eq:LN_norm1}
\begin{aligned}
& \underset{\mathbf{u}}{\text{min}}
& & \left\lVert\mathbf{r}\right\rVert_1\\
& \text{subject to} & & \mathbf{r}=\mathbf{Au} - \mathbf{b}\\
\end{aligned}
\end{IEEEeqnarray}
\noindent (\ref{Eq:LN_norm1}) can be solved efficiently using ADMM~\cite{boyd2011distributed} by reformulating it as:
\begin{IEEEeqnarray}{l}
\label{Eq:LN_norm1_ADMM}
\begin{aligned}
& \underset{\mathbf{u,z}}{\text{min}}
& & \left\lVert\mathbf{z}\right\rVert_1 \\
& \text{subject to} & & \mathbf{Au} - \mathbf{z}=  \mathbf{b} \\
\end{aligned}
\end{IEEEeqnarray}
\noindent Using the standard ADMM steps, (\ref{Eq:LN_norm1_ADMM}) can be solved iteratively as: 
\begin{IEEEeqnarray}{rCl}
\IEEEyessubnumber*
\label{Eq:LN_norm1_ADMM_itr}
\mathbf{u}^{k+1} & = & \mathbf{G}\mathbf{A}^T\left(\mathbf{b}+\mathbf{z}^k-\frac{\mathbf{y}^k}{\rho}\right)  \\ 
\label{Eq:LN_norm1_ADMM_itr_b}
\mathbf{z}^{k+1} & = & \mathit{S}_{\frac{1}{\rho}}\left(\mathbf{A}\mathbf{u}^{k+1}-\mathbf{b}^k+\frac{\mathbf{y}^k}{\rho}\right) \\
\label{Eq:LN_norm1_ADMM_itr_c}
\mathbf{y}^{k+1} & = & \mathbf{y}^{k} + \rho\left(\mathbf{A}\mathbf{u}^{k+1} - \mathbf{z}^{k+1} - \mathbf{b} \right)
\end{IEEEeqnarray}
where $\mathbf{G} = (\mathbf{A}^T\mathbf{A})^{-1}$, $\mathbf{y}$ is the dual variable (Lagrange multiplier), and $\rho \left(>0\right)$ is the penalty parameter for the violation of the linear constraint~\cite{note_ADMM_Han}. $\mathit{S}_{\frac{1}{\rho}}(x) = \left\{\max\left(\left|x\right| - \frac{1}{\rho},0\right).\text{sign}(x)\right\}$ is the proximal operator of $\ell_1$~\cite{parikh2014proximal}. The convergence criteria for solving (\ref{Eq:LN_norm1_ADMM}) using ADMM is $\left|\lVert\mathbf{z}^{k+1}\rVert_1 - \lVert \mathbf{z}^k\rVert_1\right| \leq Conv_{\text{ADMM}}$, where $Conv_{\text{ADMM}} > 0$. 

We propose two localization techniques LN-1 and LN-1E. LN-1 attempts to solve the optimization problem (\ref{Eq:LN_norm1_ADMM}) in order to localize the target node in both uncoordinated and coordinated attack scenarios. LN-1E is an improvement over LN-1 to handle only coordinated attacks. LN-1E identifies the malicious anchor nodes by applying LN-1 and then recomputes the plane by solving (\ref{Eq:LN_norm1_ADMM}) after eliminating the measurements corresponding to the malicious anchor nodes. Using LN-1, we determine the ``best fitting" plane for the measurements from all the anchor nodes $\hat{f}(x,y)=\widehat{\alpha}x + \widehat{\beta}y+ \widehat{\gamma}$. It is found that the data points corresponding to the non-malicious anchor nodes are closer to this plane than those corresponding to the malicious anchor nodes. K-means clustering~\cite{Bishop_PRM} is used to partition the anchor nodes into two groups based on their distance from this plane. The steps for implementing LN-1 and LN-1E are listed in Algorithm~\ref{Algo:LN-1E}. The node-to-plane distances are grouped into two clusters using the K-means clustering algorithm, and the centroid of the cluster containing the malicious nodes is always farther from the plane than the centroid of the other cluster. After identifying the two clusters, we recompute the ``best fitting" plane for the measurements only from the non-malicious anchor nodes as $\hat{f}_{\text{new}}(x,y)=\widehat{\alpha}_{\text{new}}x + \widehat{\beta}_{\text{new}}y+ \widehat{\gamma}_{\text{new}}$.
\begin{algorithm}
\caption{LN-1 and LN-1E Localization}
\label{Algo:LN-1E}
\begin{algorithmic}[1]
\renewcommand{\algorithmicrequire}{\textbf{Input:}}
\renewcommand{\algorithmicensure}{\textbf{Output:}}
\REQUIRE $\mathbf{A}$, $\mathbf{b}$, $N$
\ENSURE  Estimated target locations $\mathbf{\hat{t}}_{LN1}$,  $\mathbf{\hat{t}}_{LN1E}$
% \\ \underline{Initialization:} 
 %$\boldsymbol{\mathbf{A}}$, $\mathbf{b}$, 
% \textit{K} = $\emptyset$ 
\\ \underline{Computing the ``best fitting" plane using measurements from all the anchor nodes:}
\STATE $\mathbf{\widehat{u}} = \left[ \widehat{\alpha} \ \widehat{\beta} \  \widehat{\gamma} \right]^T$ by solving (\ref{Eq:LN_norm1_ADMM})
\STATE $\hat{f}(x,y)=\hat{\alpha}x+\hat{\beta}y+\hat{\gamma}$
\STATE $\mathbf{\hat{t}}_{LN1} = \left[ \widehat{\alpha} \ \widehat{\beta} \right]^{T}$
\STATE Compute distances $d_i^p$ of the data points $\langle-2a_i^x, -2a_i^y, \overline{d_i}^2- \left(a_i^x\right)^2-\left(a_i^y\right)^2 \rangle$ from the plane $\hat{f}(x,y),i=1,...,N$.
\STATE Using K-means algorithm, group the distances $d_i^p$ into two clusters. The cluster closer to the plane $\hat{f}(x,y)$ is assumed to represent the non-malicious anchor nodes.
\STATE $\hat{f}_{\text{new}}(x,y)=\widehat{\alpha}_{\text{new}}x + \widehat{\beta}_{\text{new}}y+ \widehat{\gamma}_{\text{new}}$, using only measurements from non-malicious anchor nodes.
\STATE $\mathbf{\widehat{u}}_{\text{new}} = \left[ \widehat{\alpha}_{\text{new}} \ \widehat{\beta}_{\text{new}} \  \widehat{\gamma}_{\text{new}} \right]^T$ by solving (\ref{Eq:LN_norm1_ADMM})
\STATE $\mathbf{\hat {t}}_{LN1E} = \left[ \widehat{\alpha}_{\text{new}} \ \widehat{\beta}_{\text{new}} \right]^{T}$
% \STATE $\hat{f}_{\text{mal}}(x,y)=\widehat{\alpha}_{\text{mal}}x + \widehat{\beta}_{\text{mal}}y+ \widehat{\gamma}_{\text{mal}}$, using only measurements from malicious anchor nodes.
% \STATE $\mathbf{\widehat{u}}_{\text{mal}} = \left[ \widehat{\alpha}_{\text{mal}} \ \widehat{\beta}_{\text{mal}} \  \widehat{\gamma}_{\text{mal}} \right]^T$ by solving (Eq. 11) %(\ref{Eq:LN_norm1_ADMM})
% \STATE $\mathbf{\hat {t}}_{\text{att}} = \left[ \widehat{\alpha}_{\text{mal}} \ \widehat{\beta}_{\text{mal}} \right]^{T}$
\end{algorithmic} 
\end{algorithm}
\begin{remark}
The identification of malicious anchor nodes will also enable determining the location $\mathbf{t}_{\text{att}}$ where the malicious nodes intend to make the target appear to be located. This can be useful from a security perspective in many applications and will be explored as part of future work.
\end{remark}
\section{Cramer-Rao lower bound (CRLB)}
\label{Sec:crlb}
CRLB provides a lower bound on the variance of an unbiased estimator and can be used as a benchmark for other estimators~\cite{kay1993fundamentals,crlb_Ouyang}. The mean square error (MSE) of an unbiased estimator $\left(\mathbf{\hat{t}}\right)$ of the target node position $\left(\mathbf{t}\right)$ can be expressed as shown below and bounded using the CRLB as:
\begin{IEEEeqnarray}{rCl}
\text{MSE}\left(\mathbf{\hat{t}}\right) &=& \mathbb{E}\left[\left(\hat{t^x}-{t^x}\right)^2\right] + \mathbb{E}\left[\left(\hat{t^y}-{t^y}\right)^2\right]=\text{Var}\left(\hat{t^x}\right) + \text{Var}\left(\hat{t^y}\right)\geq \left[\mathbf{F}^{-1}\right]_{11}+\left[\mathbf{F}^{-1}\right]_{22}=\text{tr}\left(\mathbf{F}^{-1}\right) 
\nonumber
\label{eq:mse_crlb}
\end{IEEEeqnarray}
where $\mathbf{F}$ is the Fisher information matrix (FIM). The RMSE of an unbiased estimator satisfies     
%\begin{IEEEeqnarray}{rCll}
$\text{RMSE}\left(\mathbf{\hat{t}}\right) \geq  \sqrt{\text{tr}\left(\mathbf{F}^{-1}\right)}$.
%\label{eq:rms_crlb}
%\end{IEEEeqnarray}
Thus, the CRLB provides a lower bound on the RMSE of unbiased estimators for estimating the target node position.   
 
For computing the CRLB, we assume that the identity of malicious and non-malicious anchor nodes is known. In addition, $\sigma$ and $\sigma_{\text{att}}$ are assumed to be known. 
%The coordinates of all the anchors in the topology are given by matrix $\mathbf{A}=[\mathbf{a_1} \ \mathbf{a_2} \ \hdots \ \mathbf{a_N}]^T$
We define $\mathcal{A}_{\text{nm}}$ and $\mathcal{A}_{\text{m}}$ as sets containing the indices of non-malicious and malicious anchor nodes, respectively.
%i.e. $\mathcal{A}_{\text{nm}}=\left\{k\,\vert\text{ anchor } k \text{ is non-malicious}\right\}$ and  $\mathcal{A}_{\text{m}}=\left\{k\,\vert\text{ anchor } k \text{ is malicious}\right\}$.
%
\subsection{CRLB for Uncoordinated Attack}
The probability density function (PDF) of the received signal power at the target node in an uncoordinated attack can be expressed as:   
\begin{IEEEeqnarray}{c}
\label{eq:crlb_probability_non_coordinated}
p(\mathbf{P}^\text{r};\mathbf{t})=\prod_{j=1}^{P}\left[\prod_{i \in \mathcal{A}_{\text{nm}}} \frac{1}{\sqrt{2 \pi \sigma^2}}\exp\frac{ -\left( p^r_{ij}-p_0+10n\log_{10}(d_i) \right)^2}{2 \sigma^2}\right.\nonumber\\ 
\times \left.\prod_{k\in \mathcal{A}_{\text{m}}}\frac{1}{\sqrt{2\pi\sigma^2_{\text{eff}}}}\exp\frac{-\left(p^r_{kj}-p_0+10n\log_{10}(d_k)\right)^2}{2\sigma^2_{\text{eff}}}\right]
%p(\mathbf{P}^\text{r};\mathbf{t})=\frac{1}{\sqrt{2\pi}}\prod_{j=1}^{P}\left[\prod_{i \in \mathcal{A}_{\text{nm}}} \frac{1}{\sigma}\exp\frac{ -\left( p^r_{ij}-p_0+10n\log_{10}(d_i) \right)^2}{2 \sigma^2}\right.\left.\prod_{k\in \mathcal{A}_{\text{m}}}\frac{1}{\sigma_{\text{eff}}}\exp\frac{-\left(p^r_{kj}-p_0+10n\log_{10}(d_k)\right)^2}{2\sigma^2_{\text{eff}}}\right] \nonumber
\end{IEEEeqnarray}
where $\sigma_{\text{eff}}^{2}=\sigma^2 + \sigma_{\text{att}}^2$. %$d_i=\left\lVert\mathbf{a}_i - \mathbf{t}\right\rVert_2$ and $d_k=\left\lVert\mathbf{a}_k - \mathbf{t}\right\rVert_2$. 
The PDF $p(\mathbf{P}^\text{r};\mathbf{t})$ satisfies the regularity conditions $\mathbb{E}\left[\frac{\partial\ln(p(\mathbf{P}^\text{r};\mathbf{t}))}{{\partial t^x}} \right]=0$ and $\mathbb{E}\left[\frac{\partial\ln(p(\mathbf{P}^\text{r};\mathbf{t}))}{{\partial t^y}}\right]=0$, and therefore, the CRLB is given by  ${t^{\text{uc}}_{\text{CRLB}}}=\sqrt{\text{tr}(\mathbf{F}_{\text{uc}}^{-1})}$ where the FIM is given by $\mathbf{F}_{\text{uc}}=\left[f^{\text{uc}}_{xx} \  f^{\text{uc}}_{xy} ; f^{\text{uc}}_{yx} \ f^{\text{uc}}_{yy}\right]$ where:
\begin{IEEEeqnarray}{rCll}
\label{eq:CRLB_non_coordinated}
f^{\text{uc}}_{xx} & = & \frac{100 P n^2}{\ln^{2}(10)}\Bigg[ & \frac{1}{\sigma^2}\sum_{i\in \mathcal{A}_{\text{nm}}}\frac{ (a_i^x-t^x)^2}{\lVert\mathbf{a}_i-\mathbf{t}\rVert_2^4} + \frac{1}{\sigma^2_{\text{eff}}}\sum_{k\in \mathcal{A}_{\text{m}}}\frac{ (a_k^x-t^x)^2}{\lVert\mathbf{a}_k-\mathbf{t}\rVert_2^4}\Bigg]\nonumber\\     
f^{\text{uc}}_{yy} & = & \frac{100 P n^2}{\ln^{2}(10)}\Bigg[ & \frac{1}{\sigma^2}\sum_{i\in \mathcal{A}_{\text{nm}}}\frac{ (a_i^y-t^y)^2}{\lVert\mathbf{a}_i-\mathbf{t}\rVert_2^4} + \frac{1}{\sigma^2_{\text{eff}}}\sum_{k\in \mathcal{A}_{\text{m}}}\frac{ (a_k^y-t^y)^2}{\lVert\mathbf{a}_k-\mathbf{t}\rVert_2^4}\Bigg]\\    
f^{\text{uc}}_{xy} & = & \frac{100 P n^2}{\ln^{2}(10)}\Bigg[ & \frac{1}{\sigma^2}\sum_{i\in \mathcal{A}_{\text{nm}}}\frac{(a_i^x-t^x)(a_i^y-t^y)}{\lVert\mathbf{a}_i-\mathbf{t}\rVert_2^4} +\> \frac{1}{\sigma^2_{\text{eff}}} \sum_{k \in \mathcal{A}_{\text{m}}}\frac{ (a_k^x-t^x)(a_k^y-t^y)}{\lVert\mathbf{a}_k-\mathbf{t}\rVert_2^4}\Bigg]\nonumber  
\end{IEEEeqnarray}
\subsection{CRLB for Coordinated Attack}
The PDF of the measurement matrix $\mathbf{P}^\text{r}$ in a coordinated attack can be expressed as:   
\begin{IEEEeqnarray}{c}
%\mathit{p}(\mathbf{P}^\text{r};\mathbf{t}) = \prod_{j=1}^{P} \Bigg[ \prod_{i \in \mathcal{A}_{\text{nm}}} \frac{1}{\sqrt{2 \pi \sigma^2}} exp \frac{-\left(p^r_{ij} - p_0 + 10n\log_{10}(d_i) \right)^2  }{2 \sigma^2}    \nonumber \\ 
%\times \prod_{k \in \mathcal{A}_{\text{m}}} \frac{1}{\sqrt{2 \pi \sigma^2}} exp  \frac{-\left(p^r_{kj} - p_0 + 10n\log_{10}(d_k^{\prime}) \right)}{2 \sigma^2}^2 \Bigg]\\
\mathit{p}(\mathbf{P}^\text{r};\mathbf{t}) = \frac{1}{\sqrt{2 \pi \sigma^2}} \prod_{j=1}^{P} \Bigg[ \prod_{i \in \mathcal{A}_{\text{nm}}}  \exp \frac{-\left(p^r_{ij} - p_0 + 10n\log_{10}(d_i) \right)^2  }{2 \sigma^2}   
 \prod_{k \in \mathcal{A}_{\text{m}}} \exp  \frac{-\left(p^r_{kj} - p_0 + 10n\log_{10}(d_k^{\prime}) \right)}{2 \sigma^2}^2
\Bigg]\nonumber
\label{eq:crlb_probability_coordinated}
\end{IEEEeqnarray}
where $d_k^{\prime}=\left\lVert\mathbf{a}_k - \mathbf{t_{\text{att}}}\right\rVert_2$. The PDF $\mathit{p}(\mathbf{P}^\text{r};\mathbf{t})$ satisfies the regularity conditions $\mathbb{E}\left[\frac{\partial\ln(\mathit{p}(\mathbf{P}^\text{r};\mathbf{t}))}{{\partial t^x}}\right] = 0$ and $\mathbb{E}\left[ \frac{\partial\ln(\mathit{p}(\mathbf{P}^\text{r};\mathbf{t}))}{{\partial t^y}}\right] = 0$, and therefore the CRLB in the coordinated attack is  $t_{\text{CRLB}}^{c}=\sqrt{\text{tr}(\mathbf{F}_\text{c}^{-1})}$. The FIM for coordinated attack is given by $\mathbf{F}_{\text{c}}=[f_{xx}^{\text{c}} \  f_{xy}^{\text{c}} ; f_{yx}^{\text{c}} \ f_{yy}^{\text{c}}]$ where:
\begin{IEEEeqnarray}{rCll}
\label{eq:CRLB_coordinated}
f_{xx}^{\text{c}} & = &  \frac{100 P n^2}{\sigma^2\ln^{2}(10)} \Bigg[&   \sum_{i \in \mathcal{A}_{\text{nm}}} \frac{ (a_i^x-t^x)^2}{\left\lVert\mathbf{a}_i-\mathbf{t}\right\rVert_2^4}  
          + \sum_{k \in \mathcal{A}_{\text{m}}} \frac{ (a_k^x-t^x_{\text{att}})^2}{\left\lVert\mathbf{a}_k-\mathbf{t_{\text{att}}}\right\rVert_2^4}\Bigg] \nonumber\\
f_{yy}^{\text{c}} & = &  \frac{100 P n^2}{\sigma^2\ln^{2}(10)} \Bigg[& \sum_{i \in \mathcal{A}_{\text{nm}}} \frac{ (a_i^y-t^y)^2}{\left\lVert\mathbf{a}_i-\mathbf{t}\right\rVert_2^4}  
        +  \sum_{k \in \mathcal{A}_{\text{m}}} \frac{ (a_k^y-t^y_{\text{att}})^2}{\left\lVert\mathbf{a}_k-\mathbf{t_{\text{att}}}\right\rVert_2^4}\Bigg]\\
f_{xy}^{\text{c}} & = &  \frac{100 P n^2}{\sigma^2\ln^{2}(10)} \Bigg[&  \sum_{i \in \mathcal{A}_{\text{nm}}} \frac{ (a_i^x-t^x)(a_i^y-t^y)}{\left\lVert\mathbf{a}_i-\mathbf{t}\right\rVert_2^4} + \sum_{k \in \mathcal{A}_{\text{m}}} \frac{ (a_k^x-t^x_{\text{att}})(a_k^y-t^y_{\text{att}})}{\left\lVert\mathbf{a}_k-\mathbf{t_{\text{att}}}\right\rVert_2^4}\Bigg]\nonumber
\end{IEEEeqnarray}
%
% \begin{table*}[]
% \captionsetup{justification=centering, labelsep=newline}
% \centering
% \caption{List of localization techniques being compared}
% \label{Tb:loc_tech_desc}
% \begin{tabular}{|c|l|l|}
% \hline
% S.No. & Algorithm          & Description \\ \hline
% 1     & LS       & Least Square based localization~\cite{linerarization_Yu}                                   \\ \hline
% 2     & Grad-Desc & Iterative  gradient  descent  with  selective pruning~\cite{secure_loc_garg}  \\ \hline
% 3     & LMdS      & Least Median Square~\cite{median_sq_error_li}     \\ \hline
%     4     & WLS      & Weighted LS   \\ \hline
% 5     & SWLS     & Secure Weighted LS     \\ \hline
% 6     & LN-1      & $\ell_1$-Norm based~(\ref{Eq:LN_norm1_LP})                                           \\ \hline
% 8     & ML-RAND   & ML estimator~(\ref{eq:min_ML}), with random initialization                      \\ \hline
% 9     & ML-LS    & ML estimator~(\ref{eq:min_ML}), initialized with the result of LS                    \\ \hline
% 10     & ML        & ML estimator~(\ref{eq:min_ML}), initialized with true target node location \\ \hline
% \end{tabular}
% \end{table*}
%
\begin{table*}[]
\captionsetup{justification=centering, labelsep=newline}
\centering
\caption{Different localization techniques being compared.}
\label{Tb:loc_tech_desc}
\begin{tabular}{|l|l|c|c|}
\hline
\multicolumn{1}{|c|}{Algorithm} & \multicolumn{1}{c|}{Description} & Uncoordinated & Coordinated\\ 
\hline
LS & Least Square based localization~\cite{linerarization_Yu} & \cmark & \xmark\\
\hline
Grad-Desc & Iterative gradient descent with selective pruning~\cite{secure_loc_garg} & \cmark & \cmark\\ 
\hline
LMdS & Least Median Square~\cite{median_sq_error_li} & \cmark & \cmark\\ 
\hline
WLS & Weighted LS & \cmark & \cmark\\ 
\hline
SWLS & Secure Weighted LS & \cmark & \xmark\\ 
\hline
LN-1 & $\ell_1$-Norm based & \cmark & \cmark\\ 
\hline
LN-1E & $\ell_1$-Norm based with malicious anchor node elimination & \xmark & \cmark\\ 
\hline
%ML-RAND   & ML estimator~(\ref{eq:min_ML}), with random initialization    & \cmark & \xmark\\ 
%\hline
%ML-LS & ML estimator~(\ref{eq:min_ML}), initialized with the result of LS   & \cmark & \xmark\\ \hline
ML & ML estimator~(\ref{eq:min_ML}), initialized with true target node location & \cmark & \xmark\\ 
\hline
\end{tabular}
\end{table*}
The estimators discussed in this paper (refer Table~\ref{Tb:loc_tech_desc}) are biased (determined via simulation) and thus the CRLB cannot be used to lower bound their performance. However, the CRLB is used as a benchmark as it represents the minimum RMSE that can be achieved by an unbiased estimator.  
\section{Performance evaluation}\label{Sec:perf_eval}
We have carried out extensive performance evaluation of the proposed secure localization techniques and compared it with four existing techniques from the literature, namely least squares (LS)~\cite{kay1993fundamentals}, LMdS~\cite{median_sq_error_li}, gradient descent (Grad-Desc)~\cite{secure_loc_garg,grad_ravi_garg}, and maximum likelihood (ML)~\cite{kay1993fundamentals} methods. 
%We kept the simulation parameters almost similar to the one present in~\cite{secure_loc_garg} for comparison with their technique. 
We consider a network spread over a $100\text{ m}\times 100$ m area with $29$ anchor nodes and one target node. Anchor nodes are set to transmit at $-10$ dBm ($p_0$), and the path loss exponent ($n$) is assumed to be $4$ representing a suburban environment~\cite{path_loss}. 
%The performance of the localization techniques does not depend on the value of the path loss exponent. 

The LS method solves for the target node location as $\hat{\mathbf{t}}=\left(\mathbf{A}^{T}\mathbf{A}\right)^{-1}\mathbf{A}^{T}\mathbf{b}$. ML estimate for the target node location is obtained by solving the non-convex optimization problem:
\begin{IEEEeqnarray}{c}
\mathbf{\hat{t}} = \underset{\mathbf{t}}{\text{argmin}} \sum_{i=1}^{N}\sum_{j=1}^{P}\left(p^r_{ij}-p_0 + 10n\log_{10}\left(d_i\right)\right)^2
\label{eq:min_ML}
\end{IEEEeqnarray}
We solve~\eqref{eq:min_ML} using the \textit{fminunc} function in Matlab which is based on the quasi-Newton method~\cite{broyden1970quasi}. %We consider two variations of the ML method, ML-LS and ML, which differ in terms of the initialization for the optimization procedure. 
%In \textit{ML--RAND} the starting point is randomly selected for each iteration, \textit{ML-LS} is initialized by the results obtained from LMS, and ML begins its first iteration using the actual location of the target node.
ML is initialized with the true location of the target node. 
%\textcolor{blue}{For implementing SWLS the values of $\zeta$, w, and $\sigma_{est}^{\text{max}}$ are set to 1.5, 0.01, and 11 respectively. $\zeta$ and w are set heuristically and the value of $\sigma_{est}^{\text{max}}$ depends on the maximum value of $\sigma$ and $\sigma_{\text{att}}$ which are 4~dB and 10~dB for the simulation carried out in the research.} 
For LMdS, we consider $20$ (intersecting) subsets with each subset consisting of $4$ anchor nodes. For Grad-Desc, the maximum number of iterations is $200$ and a constant step size of $0.4$ is chosen. Variable step size is not considered as its performance is reported to be similar to that with a constant step size~\cite{secure_loc_garg}. The threshold for the anchor selection or pruning step in Grad-Desc is empirically set to a value that gives the best performance. In SWLS, $\zeta$ is empirically set to 1.5. LN-1 and LN-1E use ADMM to solve (\ref{Eq:LN_norm1_ADMM}), and the value of $Conv_{\text{ADMM}}$ and $\rho$ are determined empirically and set to $10^{-6}$ and 0.2, respectively. The maximum number of iterations for allowing ADMM to converge is set to 5000. All results reported in this paper are based on $5000$ Monte Carlo simulations.
\subsection{Uncoordinated Attack}
Fig.~\ref{fig:topology_nc_random_anchor} shows a randomly deployed network with $29$ anchor nodes and one target node. The network is assumed to contain $8$ malicious anchor nodes (i.e., roughly $28\%$ of the anchor nodes are malicious) which attempt to disrupt the target node's localization process via uncoordinated attack. Fig.~\ref{fig:nc_rms_2db} and Fig.~\ref{fig:nc_rms_2db_target_edge} show the performance of the various secure localization techniques in terms of the RMS localization error. The RMS localization error is shown as a function of $\sigma_{\text{att}}$. %for different measurement noise ($\sigma$) levels. 
The target node estimates its position by executing the localization process after receiving $10$ packets from each of the anchor nodes. In the Monte Carlo simulations, the topology and the percentage of malicious anchor nodes are kept fixed, while the malicious anchor nodes are chosen randomly from the $29$ anchor nodes for each simulation run. Simulations are also carried out to study the performance of the algorithms as the target node moves closer to the edge of the network (refer Fig.~\ref{fig:nc_rms_2db_target_edge}). 
%
%In Fig.~\ref{fig:nc_rms_2db_target_edge} the target node is located to the edge of the network at a coordinate (10,60).      
\begin{figure*}[!b]
\centering
\begin{subfigure}[]{\textwidth}
\centering
\includegraphics[width=10cm,keepaspectratio]{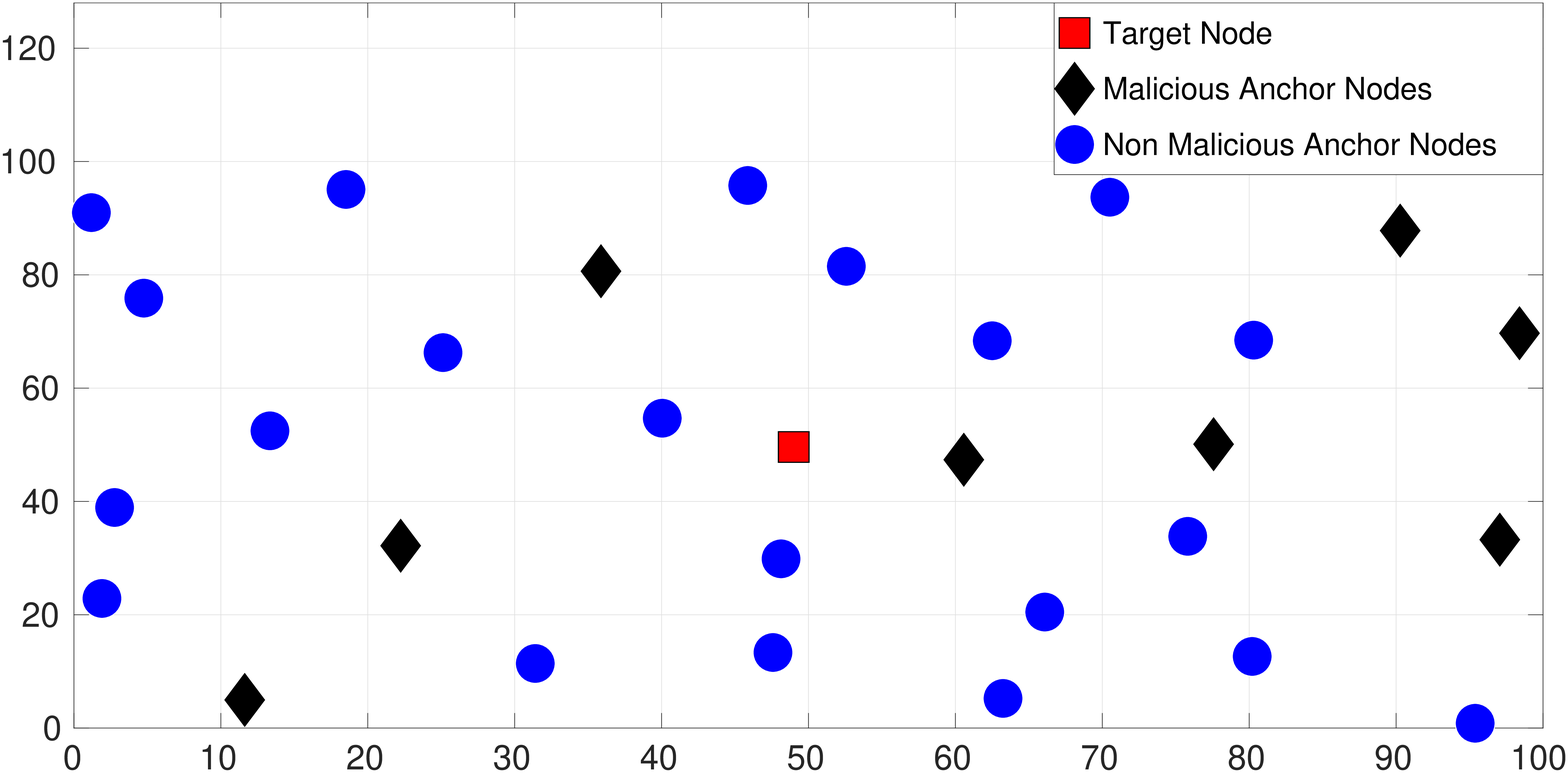}
\caption{Network topology for an uncoordinated attack with randomly placed anchors nodes (Percentage of malicious anchor nodes is $28\%$).}
\label{fig:topology_nc_random_anchor}
\end{subfigure}
\begin{subfigure}[]{\textwidth}
\centering
\includegraphics[width=12cm,keepaspectratio]{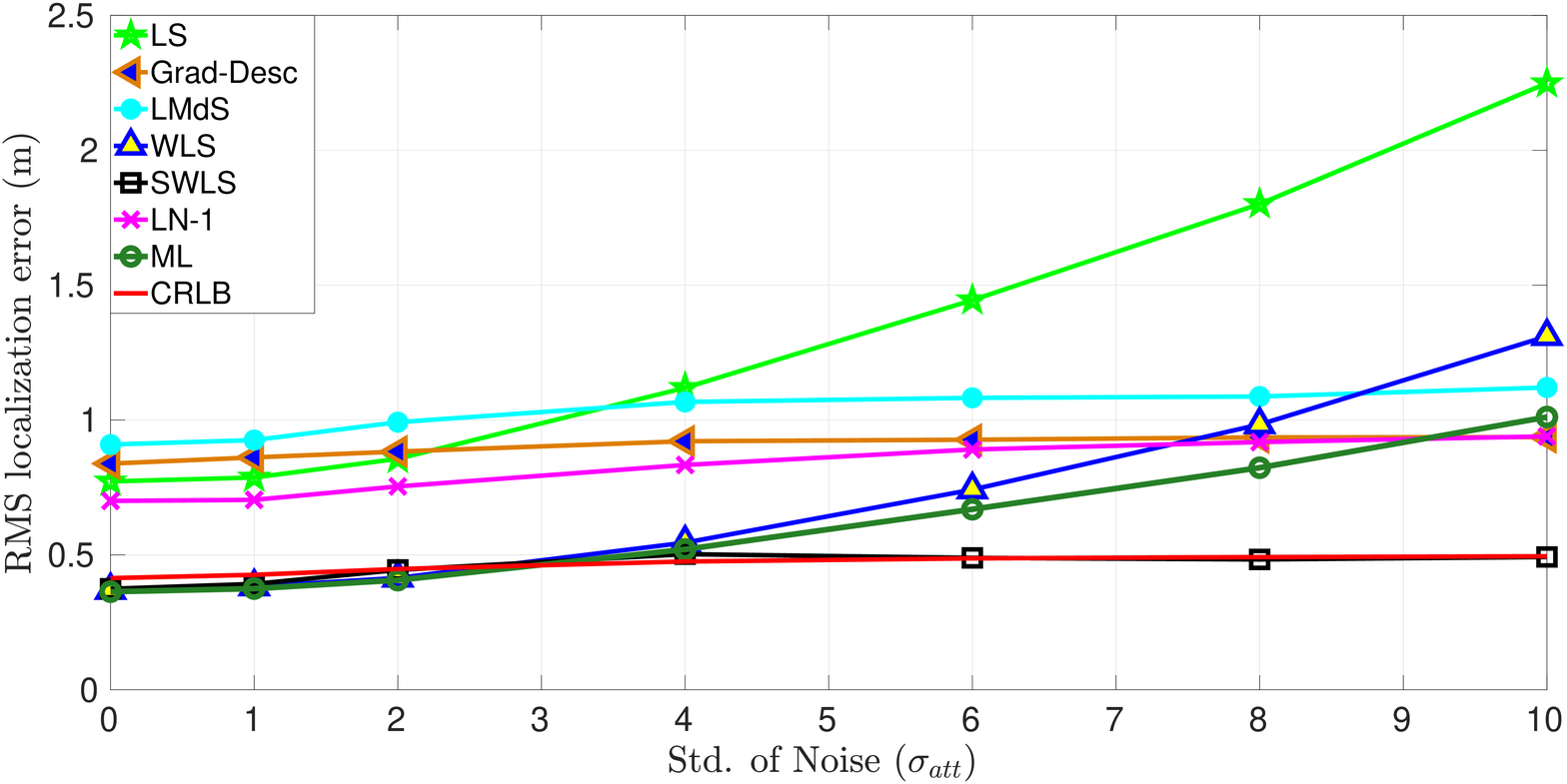}
\caption{RMS localization error of the secure localization techniques with $\sigma$ = 2 dB as a function of $\sigma_{\text{att}}$.      }
\label{fig:nc_rms_2db}
\end{subfigure}
\end{figure*}
\begin{figure*}[]
\ContinuedFloat
% \begin{subfigure}[]{\textwidth}
% \centering
% \includegraphics[width=12cm,keepaspectratio]{pictures/non_coperative_rms_noise_std_sys_noise_4_crlb.eps}
% \caption{ RMS localization error of the secure localization techniques with $\sigma$ = 4 dB as a function of $\sigma_{\text{att}}$.  }
% \label{fig:nc_rms_4db}
% \end{subfigure}
%
%\begin{subfigure}[b]{0.48\textwidth}
%\includegraphics[width=12cm,keepaspectratio]{pictures/non_coperative_rms_noise_std_sys_noise_6.eps}
%\caption{}
%\label{}
%\end{subfigure}
%
\begin{subfigure}[]{\textwidth}
\centering
\includegraphics[width=12cm,keepaspectratio]{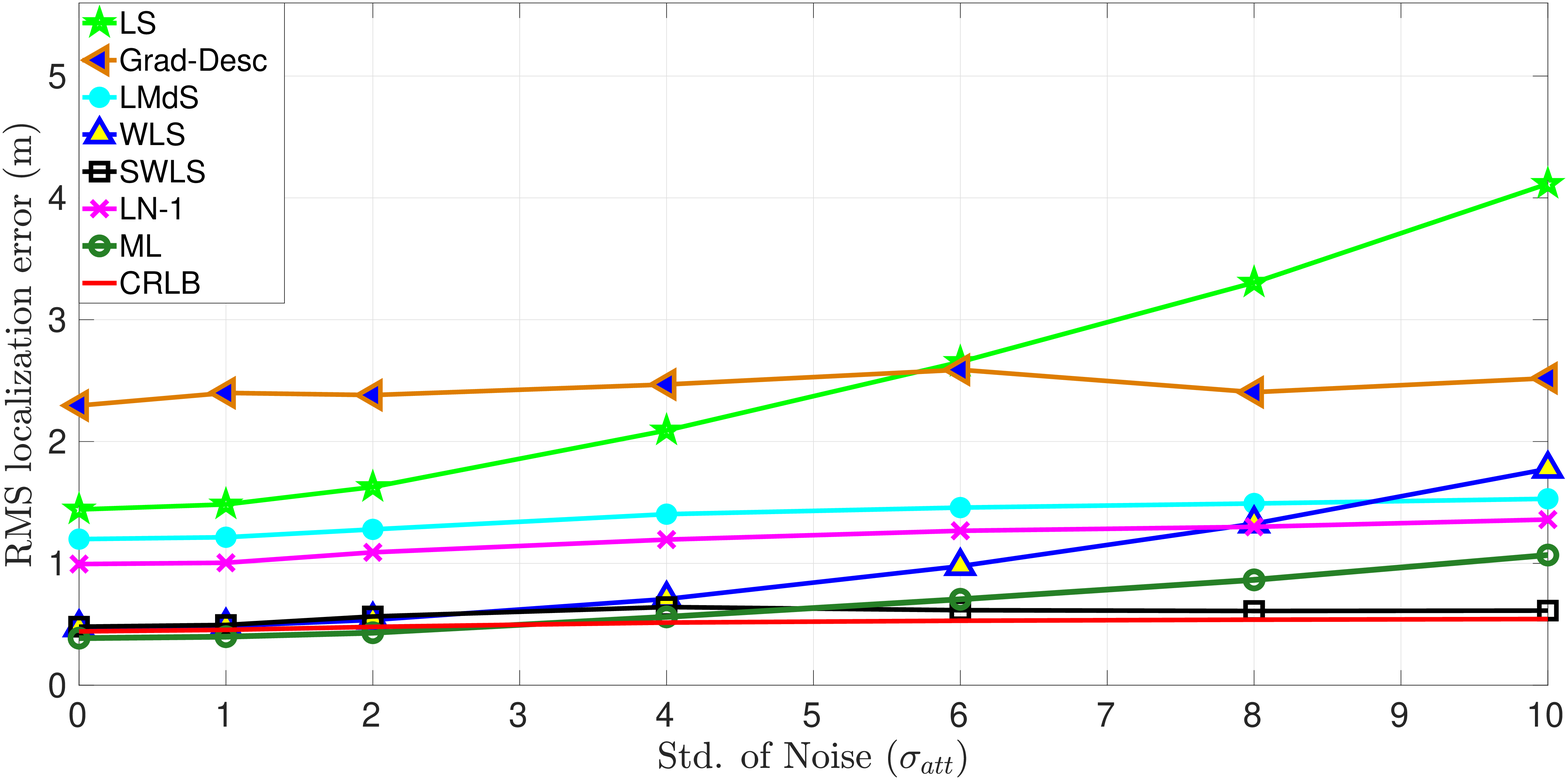}
\caption{RMS localization error as a function of $\sigma_{\text{att}}$ with $\sigma=2$ dB and the target node at location $(10,60)$.}
\label{fig:nc_rms_2db_target_edge}
\end{subfigure}
%
%\begin{subfigure}[b]{0.48\textwidth}
%\centering
%\includegraphics[width=12cm,keepaspectratio]{pictures/non_coperative_topology_uniform.eps}
%\caption{ Topology for non-coordinated attack where malicious anchors are uniformly placed (Percentage of malicious anchors is $25\%$).}
%\label{fig:topology_nc_uniform_anchor}
%\end{subfigure}
%%
%\begin{subfigure}[b]{0.48\textwidth}
%\includegraphics[width=12cm,keepaspectratio]{pictures/non_coperative_rms_noise_std_sys_noise_2_uniform.eps}
%\caption{RMS error of the secure localization techniques at $\sigma$ = 2 dB and $\sigma_{att}$ is varied (Malicious anchors uniformly distributed). }
%\label{fig:nc_rms_uniform_2db}
%\end{subfigure}
%
\caption{Performance of secure localization techniques under uncoordinated attack with random uniformly distributed anchor nodes and $P=10$ packets.}
\label{fig:nc_uniform_random_anchor}
\end{figure*}

SWLS and WLS techniques assign large weights to the anchor nodes located closer to the target node and thus reduce the effect of the larger distance estimation errors from the farther anchor nodes in the localization process. The performance of SWLS and WLS are similar as long as $\sigma_{\text{att}}$ is not significantly higher than the measurement noise ($\sigma$). However, WLS performance deteriorates as $\sigma_{\text{att}}$ becomes higher than $\sigma$ and WLS starts assigning larger weight to the malicious anchor nodes which are close to the target node. On the other hand, SWLS outperforms the other techniques as it attempts to eliminate the malicious anchor nodes from the localization process. SWLS is the only estimator whose RMS localization error is closest to the CRLB in most of the scenarios. CRLB is barely affected with the increase in the value of $\sigma_{\text{att}}$ as only $28\%$ of the anchor nodes are malicious.
%, however it is affected by an increase in $\sigma$. 
From (\ref{eq:CRLB_non_coordinated}), it can be seen that the CRLB will be significantly effected when the percentage of malicious anchor nodes is higher and $\sigma_{\text{att}}$ is greater than $\sigma$. When the target node is moved closer to the edge of the network (refer Fig.~\ref{fig:nc_rms_2db_target_edge}), the performance of Grad-Desc degrades significantly as the gradient descent algorithm appears to get stuck in a local minima. The relative performance of the other techniques is similar to the case when the target node was located at the center of the network.      

%malicious anchors can significantly affect CRLB when they are large in numbers and $\sigma_{\text{att}}$ is greater than $\sigma$. The poor performance of Grad-Desc in the scenario where the target node is located at the edge of the network is due to the gradient descent algorithm (refer Fig.~\ref{fig:nc_rms_2db_target_edge}). The algorithm gets stuck in a local minima as the objective function is not convex. The relative performance of the rest of the estimators remains almost the same with respect to the scenario where the target node is located at the center of the network.  
%
%$LMS$ and $WLMS$ (for $\sigma$ equals to 2~dB and 3~dB) are not robust to malicious anchors attack as their RMS error increases with the increase of $\sigma_{att}$ and both of them follows similar trend of RMS error. Unlike the other techniques they take into account readings from all the anchors and uses least square for location estimation.           
%
\begin{figure}
\centering
%
% \begin{subfigure}[b]{0.48\textwidth}
% \centering
% \includegraphics[width=\linewidth]{pictures/non_coperative_topology_mal_near.eps}
% \caption{Topology for uncoordinated attack where all malicious anchors are present with in 32~m radius of the target node.}
% \label{fig:topology_nc_mal_anchor_32}
% \end{subfigure}
%
\begin{subfigure}{\textwidth}
\centering
\includegraphics[width=12cm,keepaspectratio]{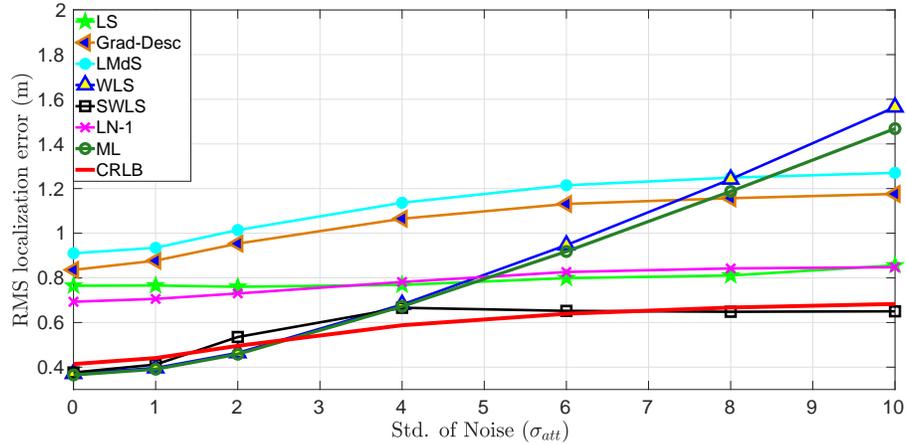}
\caption{RMS localization error when the malicious anchor nodes are close to the target node (within $32$~m).}
\label{fig:nc_rms_mal_anchor_32}
\end{subfigure}
%
%
% \begin{subfigure}[b]{0.48\textwidth}
% \centering
% \includegraphics[width=\linewidth]{pictures/non_coperative_topology_mal_mid.eps}
% \caption{Topology for uncoordinated attack where all malicious anchors are present in between 32~m and 45~m radius of the target node.}
% \label{fig:topology_nc_mal_anchor_32_45}
% \end{subfigure}
%
% \begin{subfigure}[b]{0.48\textwidth}
% \includegraphics[width=\linewidth]{pictures/non_coperative_rms_noise_sys_noise_2_mal_mid_crlb.eps}
% \caption{RMS error of the secure localization techniques when all malicious anchors are located in annular annular ring of radii 32~m and 45~m.}
% \label{fig:nc_rms_mal_anchor_32_45}
% \end{subfigure}
%
% \begin{subfigure}[b]{0.48\textwidth}
% \centering
% \includegraphics[width=\linewidth]{pictures/non_coperative_topology_mal_far.eps}
% \caption{Topology for uncoordinated attack where all malicious anchors are present beyond 45~m radius of the target node. }
% \label{fig:topology_nc_mal_anchor_45}
% \end{subfigure}
%
\begin{subfigure}{\textwidth}
\centering
\includegraphics[width=12cm,keepaspectratio]{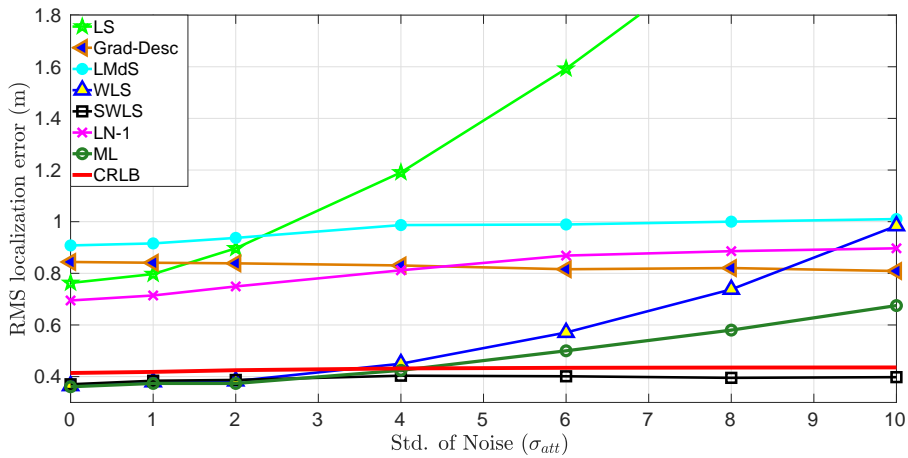}
\caption{RMS localization error when the malicious anchor nodes are located at the edge of the network.}
\label{fig:nc_rms_mal_anchor_45}
\end{subfigure}
\caption{Localization performance under uncoordinated attack when the malicious anchor nodes are moved from the center to the edge of the network. ($\sigma=2$ dB, $P=10$ packets, and percentage of malicious anchor is $28\%$.)}
\label{fig:nc_loc_mal_anc_varied}
\end{figure}
Next we study the performance of the localization techniques as the malicious anchor nodes are moved from the center towards the edge of the network. The network topology and the percentage of anchor nodes are kept the same as in Fig.~\ref{fig:nc_uniform_random_anchor}. The target node is at the center of the network and two scenarios are considered with malicious anchors located: (i) within $32$~m of the target node, and (ii) more than 45~m from the target node.
%(Fig.~\ref{fig:topology_nc_mal_anchor_32}), 
%ii) with in an annular ring of radii 32~m and 45~m, %(Fig.~\ref{fig:topology_nc_mal_anchor_32_45}), 
%(Fig.~\ref{fig:topology_nc_mal_anchor_45}). 
The localization performance for these scenarios are shown in Fig.~\ref{fig:nc_rms_mal_anchor_32} %Fig.~\ref{fig:nc_rms_mal_anchor_32_45}, 
and Fig.~\ref{fig:nc_rms_mal_anchor_45}.    
%of the system remained the same (refer Fig.~\ref{fig:topology_nc_random_anchor}) and approximately $25\%$ of the anchors were made malicious systematically, starting from the closest one to the ones at the edges. Three scenarios were created to carry out the study: (a) the malicious anchors were located within 32~m radius of the target node (Fig.~\ref{fig:topology_nc_mal_anchor_32}), (b) malicious anchors located only inside the annular ring (with target node as the center) of radius 32~m and 45~m (Fig.~\ref{fig:topology_nc_mal_anchor_32_45}) and (c) malicious node present beyond 45~m radius of the target node (Fig.~\ref{fig:topology_nc_mal_anchor_45}). The boundaries were set to 32~m and 45~m to divide the plane into three equal zones and to keep the percentage of malicious nodes fixed. Fig.~\ref{fig:nc_rms_mal_anchor_32}, Fig.~\ref{fig:nc_rms_mal_anchor_32_45}, and Fig.~\ref{fig:nc_rms_mal_anchor_45} shown the RMS error of the localization techniques in the three above mentioned scenarios with the increase of $\sigma_{\text{att}}$ when the standard deviation of the system noise was set to 2~dB and the target node performed its localization after receiving 10 packets ($P$) from each of the anchors. 
It is observed that LS is significantly affected by the position of the malicious anchor nodes and its performance deteriorates as the malicious anchor nodes move towards the edge of the network. It is also observed from Fig.~\ref{fig:nc_rms_mal_anchor_32} %and Fig.~\ref{fig:nc_rms_mal_anchor_32_45} 
that LS outperforms LMdS and Grad-Desc when malicious anchor nodes are located close to the target.
%, and in Fig.~\ref{fig:nc_rms_mal_anchor_32} LS performs better than Grad-Desc.  
LS gives equal weight to measurements from all the anchor nodes, and from Lemma~\ref{lemma:two_diff_dist_per} we know that distance estimates for the anchor nodes located farther from the target node tend to have large errors due to $\sigma_{\text{att}}$.

In Fig.~\ref{fig:nc_rms_mal_anchor_32}, we see that the localization performance of LMdS and Grad-Desc deteriorates gradually with increase in the value of $\sigma_{\text{att}}$. However, as the malicious anchor nodes move away from the target node and towards the edge of the network, these two techniques display robustness to the attack. LMdS and Grad-Desc pick four and $\frac{N}{2}$ anchor nodes, respectively. %In Fig.~\ref{fig:topology_nc_mal_anchor_32}
In Fig.~\ref{fig:nc_rms_mal_anchor_32}, as the malicious anchor nodes are close to the target node, the residuals of the subset in LMdS or their gradients in Grad-Desc are lower than when the malicious anchor nodes are farther from the target nodes resulting in some malicious anchor nodes getting picked. In contrast to LS, WLS is more robust as the malicious anchor nodes move towards the edge of the network, because it assigns large weight to the anchor nodes that are located close to the target node.
%
% In scenario 1 the malicious anchors are present near the target node and WLS is the least robust among the localization techniques (refer Fig.~\ref{fig:nc_rms_mal_anchor_32}). Even in scenario 2 and 3 WLS is affected by malicious anchors when $\sigma_{\text{att}}$ becomes grater than $\sigma$. 

SWLS outperforms the other techniques in the scenarios shown in Fig.~\ref{fig:nc_loc_mal_anc_varied}. In Fig.~\ref{fig:nc_rms_mal_anchor_32}, SWLS results in an abrupt increase in the RMS localization error as $\sigma_{\text{att}}$ becomes greater than $\sigma$. For $\sigma_{\text{att}}<\sigma$, $\sigma_{\text{eff}}$ is almost equal to or slightly greater than $\sigma$ and the localization accuracy is not significantly affected. When $\sigma_{\text{att}}>\sigma$, SWLS can identify the malicious anchor nodes and eliminate them. However when $\sigma_{\text{att}}\approx\sigma$, SWLS fails to identify the malicious anchor nodes resulting in a larger localization error as seen in Fig.~\ref{fig:nc_rms_mal_anchor_32}. Similar behaviour is expected in Fig.~\ref{fig:nc_rms_mal_anchor_45} except now the malicious anchor nodes are farther from the target and are assigned lower weights. Thus, the localization performance does not deteriorate as in Fig.~\ref{fig:nc_rms_mal_anchor_32}. ML outperforms WLS for higher values of $\sigma_{\text{att}}$, otherwise exhibits similar trend as WLS. SWLS performance in Fig.~\ref{fig:nc_loc_mal_anc_varied} is close to the CRLB and SWLS outperforms the other estimators. 

The localization performance of LN-1 is not affected by the position of the malicious anchor nodes. LN-1 performance is similar in all cases considered in Fig.~\ref{fig:nc_loc_mal_anc_varied}  as well as in the case where the malicious anchor nodes are randomly placed (refer Fig.~\ref{fig:nc_rms_2db}). LN-1 does not weight the anchor nodes differently based on the position nor eliminates any anchor nodes from the localization process. It reduces the weight of those measurements that exhibit higher variance (outliers) and thus its localization performance is affected by $\sigma_{\text{att}}$ and $\sigma$.

Fig.~\ref{fig:nc_rms_packet_per_mal_anchor} shows the result of simulations carried out to understand the effect of the number of packets ($P$) and the percentage of malicious anchor nodes on the localization performance. The network topology is the same as in Fig.~\ref{fig:topology_nc_random_anchor}. 
%In Fig.~\ref{fig:nc_rms_packet}, $\sigma$ and $\sigma_{\text{att}}$ are set to $2$~dB and $6$~dB respectively, and $25\%$ of the anchor nodes are assumed to be malicious, and the number of packets (P) is varied from 2 to 20.
Fig.~\ref{fig:nc_rms_packet} shows that with $P=2$, LN-1 and Grad-Desc outperform WLS and SWLS, however as $P$ increases WLS and SWLS outperform the other techniques. As $P$ increases, the estimated distances become more accurate resulting in a better estimate of the variance of $d^2$ (refer~\eqref{Eq:var_d^2}), and thus the performance of WLS and SWLS improves at a faster rate than the other techniques (except LS). For $P>10$, the RMS localization error of WLS and SWLS is close to ML and CRLB, respectively.
%
%It can be observed from Fig.~\ref{fig:nc_rms_packet} that with the increase in the number of packets the performance of all the techniques including the CRLB increases.  From (\ref{eq:CRLB_non_coordinated}), the Fisher matrix can be written as $P\mathbf{F^{nc}_{P=1}}$, where $\mathbf{F^{nc}_{P=1}}$ is the Fisher matrix when only a single packet is taken into consideration and $P$ is the number of packet. Therefore
% \begin{IEEEeqnarray}{l}
% t_{CRLB}^{nc} = \frac{1}{\sqrt{P}}\sqrt{tr(\mathbf{F^{nc^{-1}}_{P=1}})} \propto P^{-\frac{1}{2}}
% \label{eq:crlb_vs_packet_size}
% \end{IEEEeqnarray}
% The CRLB line in Fig.~\ref{fig:nc_rms_packet} justifies this relationship. As more number of packets are taken into consideration, prediction of distances of the anchors from the target node become more accurate resulting in better localization accuracy. 

In Fig.~\ref{fig:nc_rms_per_mal_anchor}, the localization performance of all the techniques is found to deteriorate as the percentage of malicious anchor nodes increases. From (\ref{eq:CRLB_non_coordinated}), as $\text{card}\left(\mathcal{A}_{\text{m}}\right)$ increases (and $\text{card}\left(\mathcal{A}_{\text{nm}}\right)$ decreases), the individual elements of the Fisher matrix $\mathbf{F}_{\text{uc}}$ decrease since $\sigma_{\text{eff}}>\sigma$. Thus the overall CRLB also increases with the number of malicious anchor nodes in the system. The performance of Grad-Desc, LN-1, LMdS, and SWLS begins to deteriorate at a faster rate when the percentage of malicious anchor nodes exceeds $50\%$. In this scenario, SWLS is able to eliminate the malicious anchor nodes and computes the target node position using the remaining non-malicious anchor nodes. The performance deteriorates due to the relatively fewer anchor nodes. LMdS estimates the target node position using a subset of four anchor nodes, and as the percentage of malicious anchor nodes increases, the chances of LMdS picking one or more malicious anchor nodes in its final subset also increase.
\begin{figure}[]
\centering
\begin{subfigure}[b]{\textwidth}
\centering
\includegraphics[width=12cm,keepaspectratio]{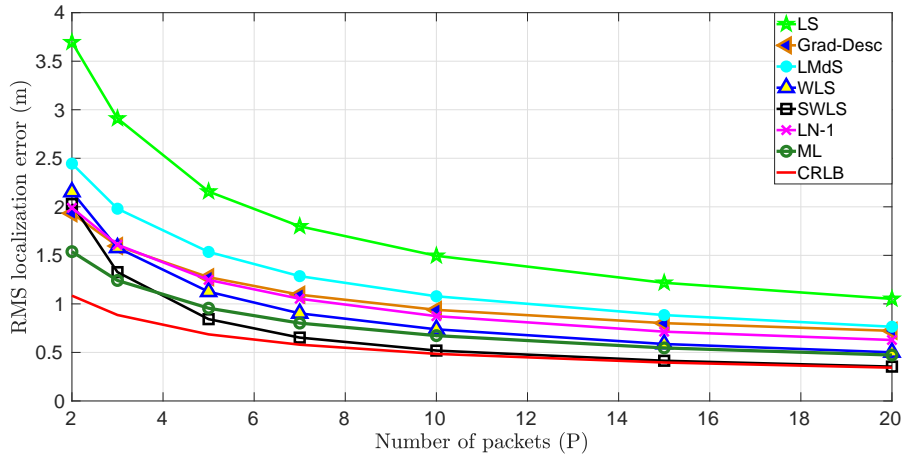}
\caption{RMS localization error as a function of the number of received packets ($P$) ($\sigma=2$~dB, $\sigma_{\text{att}}=6$~dB, and $28\%$ of the anchor nodes are malicious).}
\label{fig:nc_rms_packet}
\end{subfigure}
\begin{subfigure}[b]{\textwidth}
\centering
\includegraphics[width=12cm,keepaspectratio]{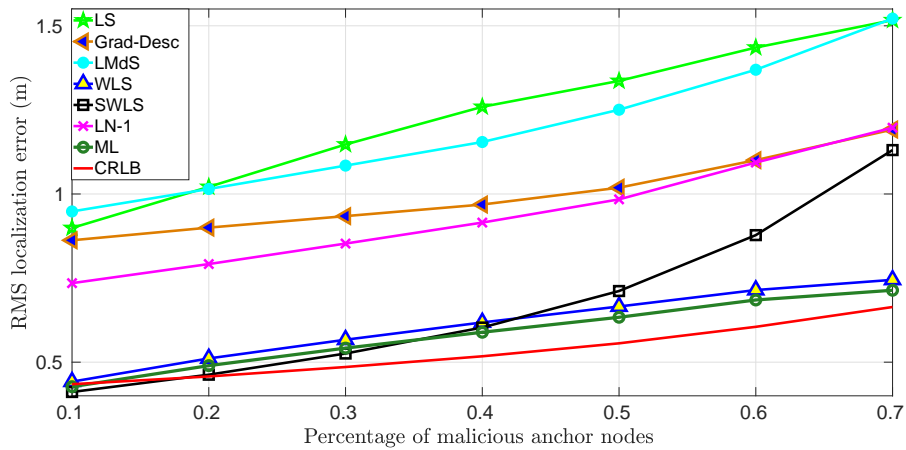}
\caption{RMS localization error as a function of the percentage of malicious anchor nodes ($\sigma=2$ dB, $\sigma_{\text{att}}=4$~dB, and $P= 10$ packets).}
\label{fig:nc_rms_per_mal_anchor}
\end{subfigure}
\caption{Performance of secure localization techniques as a function of the number of received packets and percentage of malicious anchors in uncoordinated attack.}
\label{fig:nc_rms_packet_per_mal_anchor}
\end{figure}
\subsection{Coordinated attack}
For coordinated attack, we consider the following localization techniques: WLS, Grad-Desc, LMdS, LN-1, and LN-1E. The other techniques result in poor localization performance under coordinated attack. SWLS does not perform well as it attempts to differentiate malicious anchor nodes from non-malicious ones based on the variation in the received power. This strategy fails as the malicious anchor nodes maintain a fixed transmit power in coordinated attack. Thus, we have not considered SWLS technique for the coordinated attack scenario. Fig.~\ref{fig:co_rms_dist} presents the localization performance for a coordinated attack scenario as the distance between $\mathbf{t}$ and $\mathbf{t_{\text{att}}}$ is increased. The network topology is the same as in Fig.~\ref{fig:topology_nc_random_anchor}. 
\begin{figure*}
\centering
% \begin{subfigure}[]{\textwidth}
% \centering
% \includegraphics[width=12cm,keepaspectratio]{pictures/rms_dist_noise_2_per_mal_3_coord_center}
% \caption{RMS localization error when 27\% of anchor nodes are malicious.}
% \label{fig:co_rms_dist_pec_mal_3}
% \end{subfigure}
%
\begin{subfigure}[]{\textwidth}
\centering
\includegraphics[width=10cm,keepaspectratio]{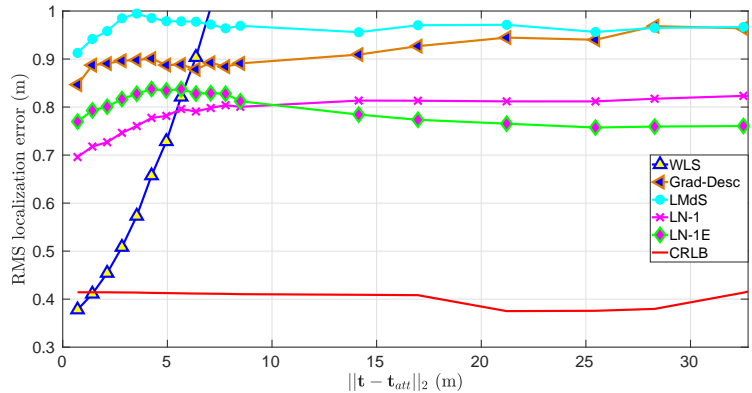}
\caption{RMS localization error when 10\% of anchor nodes are malicious.}
\label{fig:co_rms_dist_pec_mal_1}
\end{subfigure}
%\end{figure*}
%\begin{figure*}
%\ContinuedFloat
\begin{subfigure}[]{\textwidth}
\centering
\includegraphics[width=10cm,keepaspectratio]{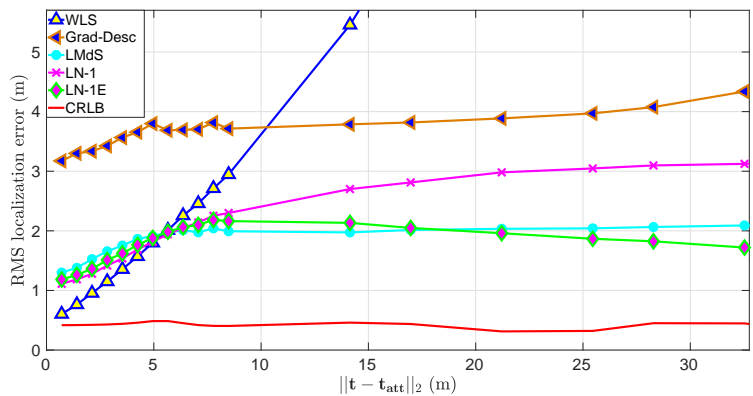}
\caption{RMS localization error when $28\%$ of anchor nodes are malicious with the target node at location $(5,50)$.}
\label{fig:co_rms_dist_pec_mal_3_target_edge}
\end{subfigure}
\caption{Performance of secure localization techniques in coordinated attack scenario ($\sigma = 2$~dB and $P=10$ packets).}
\label{fig:co_rms_dist}
\end{figure*}
%
%\textcolor{blue}{In Fig.~\ref{fig:co_rms_dist_pec_mal_3}, it is observed that Grad-Desc, LMdS, and LN-1E are robust to coordinated attacks.  Grad-Desc and LN-1E have similar performance, and outperform LMdS, LN-1 and WLS when the attack is stronger with $\left\lVert\mathbf{t}-\mathbf{t_{\text{att}}}\right\rVert_2 > $ 5~m.} When $\left\lVert\mathbf{t}-\mathbf{t_{\text{att}}}\right\rVert_2 <$ 5~m, WLS has the best performance. When the coordinated attack is mild, assigning larger weights to the anchor nodes located close to the target improves the localization performance. %The performance of LN-1E is better than LMdS and is almost similar to Grad-Desc. LN-1E results in a dramatic improvement in localization performance over LN-1 for strong attacks, as LN-1E eliminates the malicious anchor nodes from the localization process resulting in a robust algorithm.

In Fig.~\ref{fig:co_rms_dist_pec_mal_1}, with $10\%$ malicious anchor nodes, WLS significantly outperforms the other techniques when the coordinated attack is mild. However, as the attack becomes stronger the WLS performance deteriorates rapidly. LN-1 and LN-1E result in similar performance and both outperform LMdS and Grad-Desc. The poor performance of LMdS and Grad-Desc is due to the elimination of certain anchor nodes from the localization process. When the percentage of malicious anchor nodes is low and $\mathbf{t_{\text{att}}}$ is close to $\mathbf{t}$, eliminating the malicious anchor nodes does not improve the localization accuracy as the measurement noise tends to determine the localization performance. 

In Fig.~\ref{fig:co_rms_dist_pec_mal_3_target_edge}, we show the performance of the localization techniques as the target node moves closer to the edge of the network. The performance of Grad-Desc degrades significantly, while the performance of LN-1E and LMdS are similar and degrade to a lesser extent. 
%In Fig.~\ref{fig:co_rms_dist_pec_mal_1} also LN-1 outperforms Grad-Desc. It also shows similar performance with LMdS and LN-1E only when $\mathbf{t_{\text{att}}}$ is close to $\mathbf{t}$. Similar to the above discussed scenarios WLMS shows better performance with respect to the others when the value of $\left\lVert\mathbf{t}-\mathbf{t_{\text{att}}}\right\rVert_2$ is low. 
From~\eqref{eq:CRLB_coordinated}, it is noted that the FIM $\mathbf{F}_\text{c}$ does not directly depend on $\left\lVert\mathbf{t}-\mathbf{t_{\text{att}}}\right\rVert_2$, and thus the CRLB in Fig.~\ref{fig:co_rms_dist} is almost constant. It is also observed that Grad-Desc, LMdS, LN-1, and LN-1E follow a similar trend as the CRLB.
\begin{figure}
\centering
\begin{subfigure}[b]{\textwidth}
\centering
\includegraphics[width=10cm,keepaspectratio]{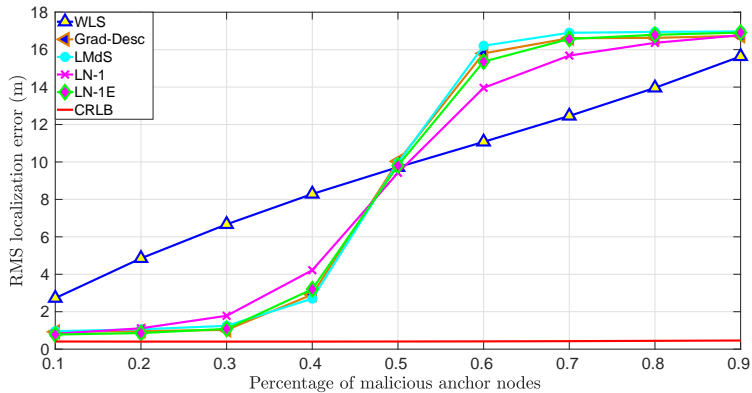}
\caption{RMS localization error with $\mathbf{t}_{\text{att}}=[t^x+12 \ t^y+12]^T$ and $\left\lVert\mathbf{t} - \mathbf{t_{\text{att}}}\right\rVert_2 = $16.97 m.}
\label{fig:co_rms_mal_anchor_dist_16}
\end{subfigure}
\begin{subfigure}[b]{\textwidth}
\centering
\includegraphics[width=10cm,keepaspectratio]{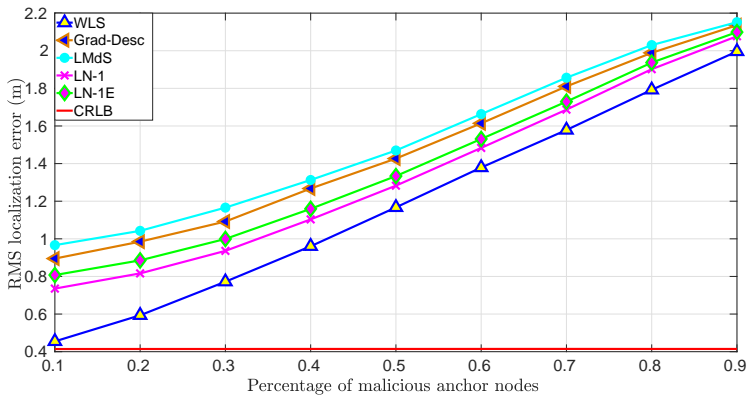}
\caption{RMS localization error with $\mathbf{t}_{\text{att}}=[t^x+1.5 \ t^y+1.5]^T$ and $\left\lVert\mathbf{t} - \mathbf{t_{\text{att}}}\right\rVert_2 = $2.12 m. }
\label{fig:co_rms_mal_anchor_dist_16_dist_2}
\end{subfigure}
\caption{Performance of the secure localization techniques in a coordinated attack as a function of the percentage of malicious anchor nodes ($\sigma = 2$~dB and $P=10$ packets).}
\label{fig:co_rms_mal_anchor}
\end{figure}

In Fig.~\ref{fig:co_rms_mal_anchor}, we study the localization performance under a coordinated attack as a function of the percentage of malicious anchor nodes. This has been simulated for two different values of $\mathbf{t}_{\text{att}}$ representing mild and severe forms of the attack. In Fig.~\ref{fig:co_rms_mal_anchor_dist_16}, it is seen that Grad-Desc, LMdS, LN-1, and LN-1E outperform WLS when the percentage of malicious anchor nodes is less than $50\%$. However, as the percentage of malicious anchor nodes goes above 50\% WLS outperforms the other techniques. Grad-Desc retains $50\%$ of the anchor nodes, and as the percentage of malicious anchor nodes exceeds $50\%$, it ends up using measurements from malicious nodes for localization. LMdS also fails to consistently find a subset containing only non-malicious anchor nodes. 
%The probability that LMdS will choose a subset containing 4 non-malicious anchor nodes when 50\% of the anchor nodes are malicious is only around 0.05. 
For LN-1 and LN-1E, the number of outliers increases with the increase in the percentage of malicious anchor nodes. This affects the robustness of $\ell_1$-norm optimization as the measurements from the malicious anchor nodes are assigned the same weight as measurements from non-malicious anchor nodes. Unlike Grad-Desc, LMdS, and LN-1E, LN-1 does not attempt to classify or cluster the measurements and instead it fits a continuous function to the measurements. Thus, when the percentage of malicious anchor nodes exceeds $50\%$, LN-1 achieves better performance than Grad-Desc, LMdS, and LN-1E. 
%The ML based techniques degrade gracefully as the percentage of malicious anchor nodes increases. It is also noted that LS outperforms WLS consistently, as the weights employed by WLS to the anchor nodes get adversely affected by the intentional changes in the transmit power introduced by the malicious anchor node. Why LS outperforms WLS:The malicious anchor nodes located far from the target node can be assigned larger weights if they change their transmit power to a level such that the target node believes the malicious anchor node to be closer than it actually is.

In Fig.~\ref{fig:co_rms_mal_anchor_dist_16_dist_2}, we consider a milder form of the coordinated attack where the malicious anchor nodes  attempt to make the target node appear at ($t^x+1.5$, $t^y+1.5$) instead of its true location ($t^x$, $t^y$). In this case, WLS outperforms the other techniques and it is followed by LN-1 and LN-1E. The performance of Grad-Desc, LMdS, and LN-1E is comparatively poor due to their anchor elimination process. From the CRLB expression for a coordinated attack~\eqref{eq:CRLB_coordinated}, it is seen that as the percentage of malicious anchor nodes increases, $\text{card}(\mathcal{A}_{\text{m}})$ increases and at the same time $\text{card}(\mathcal{A}_{\text{nm}})$ decreases, such that the CRLB remains almost constant. %\textcolor{blue}{The performance analysis of the secure localization techniques discussed in this section in uncoordinated and coordinated attack remains valid even when the anchor nodes are placed on a uniform grid.}
\subsection{Computational Complexity}
%
%
% % Please add the following required packages to your document preamble:
% % \usepackage{multirow}
% % \usepackage{graphicx}
% \begin{table}[]
% \captionsetup{justification=centering, labelsep=newline}
% \centering
% \caption{Computational complexity of the different localization techniques.}
% \label{Tb:computational_complexity}
% \resizebox{\textwidth}{!}{%
% \begin{tabular}{|c|c|c|c|c|c|c|c|c|}
% \hline
%  & Estimators & LS & WLS & SWLS & LN-1 & LN-1E & Grad-Desc & LMDS \\ \hline
% CC &  & \multicolumn{2}{c|}{$\mathcal{O}\left(N^2\right) $} & $\mathcal{O}\left( \text{card}(\mathcal{M})^2 \right)$ & $\mathcal{O}\left( \text{card}(\mathcal{M})^2 \right)$ &  $\mathcal{O}\left( \max \left( N T k_{\text{K-means}},k_{\text{ADMM}}N\right)\right)$ & $\mathcal{O}\left(k_{\text{GD}}N\right)$ & $\mathcal{O}\left(M_{\text{LMdS}}N^2\right)$ \\ \hline
% \multirow{2}{*}{ET} & Non Coordinated & 0.4 & 0.7 & 4.5 & 9-18 & - & 34.5 & 57 \\ \cline{2-9} 
%  & Coordinated & - & - &  & 7-26 & 15-28 & 31 & 52 \\ \hline
% \end{tabular}%
% }
% \end{table}
%
We next present a comparison of the computational complexity of the different localization techniques. Table~\ref{Tb:computational_complexity} shows the asymptotic complexities of the secure localization techniques considered in this work. For WLS, line 5 of Algorithm~\ref{Algo:WLMS} is the dominant computational step. 
%which solves N equations to find 3 variables and is given by $\mathbf{\hat{q}}=\left({\mathbf{A}}^T\mathbf{W}{\mathbf{A}}\right)^{-1} \mathbf{A}^T \mathbf{W} \mathbf{b}$  where $\mathbf{A} \ \in \ \mathbb{R}^{N \times 3}$, $\mathbf{W} \ \in \ \mathbb{R}^{N \times N}$, and $\mathbf{b} \ \in \ \mathbb{R}^{N \times 1}$. 
Computing $\mathbf{\hat{q}}$ requires five matrix multiplications and one matrix inversion. Therefore, the computational complexity of WLS is given by $\mathcal{O}\left(3N^2+9N+3N+3^3+3^2+N^2\right) \ \simeq \ \mathcal{O}\left({N^2}\right) \nonumber$. Similarly, the computational complexity of LS is $\mathcal{O}\left({N^2}\right)$ and SWLS is $\mathcal{O}\left(\left(\text{card}(\mathcal{M})\right)^2 \right)$, where $\mathcal{M}$ is the set of non-malicious anchor nodes identified by SWLS. 
%and $\text{card}(M) \leq N$. 
%For SWLS technique (refer Algorithm~\ref{Algo:SWLMS}) lines 4-7 and line 12 are the most dominant and have a complexity of $\mathcal{O}\left( N \frac{\sigma_{\text{est}}^{\text{max}}}{\text{w}} \right)$ and $\mathcal{O}\left({N^2}\right)$ (similar to that of WLS) respectively. The overall complexity of the algorithm is $\mathcal{O}\left( \max \left( N \frac{\sigma_{\text{est}}^{\text{max}}}{\text{w}},N^2 \right) \right)$.
LN-1 executes (\ref{Eq:LN_norm1_ADMM_itr}), (\ref{Eq:LN_norm1_ADMM_itr_b}), and (\ref{Eq:LN_norm1_ADMM_itr_c}) in an iterative manner until convergence is achieved. Assuming ADMM requires $k_{\text{ADMM}}$ iterations on average to converge, and the computational complexity of (\ref{Eq:LN_norm1_ADMM_itr}) is $\mathcal{O} \left(N\right)$, the computational complexity for the ADMM steps is $\mathcal{O} \left( k_{\text{ADMM}}N\right)$. Thus, the computational complexity of LN-1 is $\mathcal{O} \left( \max \left( k_{\text{ADMM}}N, N^2 \right)\right)  \simeq \mathcal{O} \left(  k_{\text{ADMM}}N\right)$, as $k_{\text{ADMM}} > N$ in general. In addition to the steps in LN-1, LN-1E involves K-means clustering which has a complexity of $\mathcal{O}\left( N T k_{\text{K-means}}\right)$~\cite{k_menas_time}, where $k_{\text{K-means}}$ is the number of average iterations and $T$ is the complexity for calculating the distance between two data points. The computational complexity of LN-1E is  $\mathcal{O}\left( \max \left( N T k_{\text{K-means}},k_{\text{ADMM}}N\right)\right)$. LMdS algorithm involves two main operations: (i) dividing the RSSI measurements into $M_{\text{LMdS}}$ subsets with each subset consisting of $N_{\text{LMdS}}$ anchor nodes, and the target location is estimated for each subset using LS technique, and (ii) computing median of the residue of the results obtained from each of the subsets. The computational complexity of the first operation is $\mathcal{O}\left(M_{\text{LMdS}}N_{\text{LMdS}}^2\right)$ and the second operation is $\mathcal{O}\left(M_{\text{LMdS}}N^2\right)$.
Since $N > N_{\text{LMdS}}$ in general, LMdS has a computational complexity of $\mathcal{O}\left(M_{\text{LMdS}}{N^2}\right)$. 
The computational complexity of Grad-Desc is $\mathcal{O}\left(k_{\text{GD}}{N}\right)$ where $k_{\text{GD}}$ is the total number of iterations. 
From Table~\ref{Tb:computational_complexity} it can be observed that the computational complexity of LN-1, LN-1E, and Grad-Desc varies linearly with the number of anchor nodes in the network.  
\begin{table}[]
\captionsetup{justification=centering, labelsep=newline}
\centering
\caption{Computational complexity of the different localization techniques.}
\label{Tb:computational_complexity}
\begin{tabular}{|c|c|c|c|c|c|c|}
\hline
LS & WLS & SWLS & LN-1 & LN-1E & Grad-Desc & LMdS\\ 
\hline
\multicolumn{2}{|c|}{$\mathcal{O}\left(N^2\right)$} &
$\mathcal{O}\left( \text{card}(\mathcal{M})^2 \right)$ &
$\mathcal{O} \left(  k_{\text{ADMM}}N\right)$ &  $\mathcal{O}\left( \max \left( N T k_{\text{K-means}},k_{\text{ADMM}}N\right)\right)$ & $\mathcal{O}\left(k_{\text{GD}}N\right)$ & $\mathcal{O}\left(M_{\text{LMdS}}N^2\right)$\\ 
\hline
\end{tabular}
\end{table}
\section{Conclusion}
\label{Sec:conclusion}
In this paper, we presented localization techniques that are robust in the presence of malicious anchor nodes in the network. We proposed four secure localization techniques WLS, SWLS, LN-1, and LN-1E, and compared their performance with the existing techniques Grad-Desc and LMdS. Two types of attacks were considered: uncoordinated and coordinated. All nodes in the network are assumed to transmit at a fixed power level which is known to the target node. The localization attacks are executed by the malicious anchor nodes by changing their transmit power and not reporting it to the target node. For uncoordinated attacks, the localization performance was studied with variation in the transmit power of the malicious anchor nodes, location of the malicious anchor nodes (close or far from the target node), number of packets $P$, and percentage of malicious anchor nodes in the network. SWLS outperformed all the other techniques and was close to the CRLB in most cases. The performance of WLS was found to deteriorate for $\sigma_{\text{att}}>\sigma$ or when many malicious anchor nodes are located close to the target node. On the other hand, LN-1 is neither affected by an increase in $\sigma_{\text{att}}$ nor by a change in the locations of the malicious anchor nodes. For coordinated attacks, the localization performance was studied as a function of the severity of the attack in terms of the distance between the actual and reported locations of the target node and variation in the percentage of malicious anchor nodes. LN-1 and LN-1E perform better than the other techniques with lower percentage of malicious anchor nodes and under milder form of the coordinated attack. WLS outperforms all other techniques when the attack is mild or the percentage of malicious anchor nodes exceeds $50\%$. The computational complexity of LN-1 and LN-1E varies linearly with the number of anchor nodes in the network. 
%
%Thus it can be concluded that of all the techniques discussed here LN-1 is the most robust in both uncoordinated and coordinated attack scenarios.
%
\appendices  
\section{Proof of~\eqref{Eq:close_form_sigma_est}}
\label{appendix}
%
%$f^{\text{Var}}_{\text{d}_i}(d_i,\sigma_{\text{est}})$ (refer (\ref{Eq:var_d})) is a convex function with respect to $\sigma_{\text{est}}$ (when $\sigma_{\text{est}} \geq 0$) as
Let $l(\sigma_{\text{est}}) \triangleq  v - f^{\text{var}}_{d} (\overline{d_i},\sigma_{\text{est}}) $ where $v=\text{Var}(d_{i1}, d_{i2}, \dots, d_{iP})$. It can be shown that:
%Second derivative of $f^{\text{Var}}_{\text{d}_i}(\overline{d_i},\sigma_{\text{est}})$ (refer (\ref{Eq:var_d})) with respect to $\sigma_{\text{est}}$ is
%
\begin{IEEEeqnarray}{lCl}
\label{Eq:double_differentiate}
 \frac{\partial^2 f^{\text{Var}}_{\text{d}_i}(\overline{d_i},\sigma_{\text{est}})}{\partial\sigma_{\text{est}}^2}   = 
 \underbrace{\frac{2 \overline{d_i}^2}{18.86n^2}
 \exp\left(\frac{\sigma^2_{\text{est}} }{18.86n^2}\right)}_{k_1}
 \left[  \underbrace{ 2 \exp\left(\frac{\sigma^2_{\text{est}} }{18.86n^2}\right)}_{k_2} \underbrace{\left(1+\frac{4 \sigma^2_{\text{est}}}{18.86n^2}\right)}_{k_3} - \underbrace{\left(1+\frac{2 \sigma^2_{\text{est}}}{18.86n^2}\right)}_{k_4}
 \right]\nonumber
\end{IEEEeqnarray}
%
%`In (\ref{Eq:double_differentiate}) 
where $k_2k_3-k_4$ is always positive because $k_3 \geq k_4$ and $k_2 \geq 2$. 
%$\left(\text{as} \frac{\sigma^2_{\text{est}} }{18.86n^2} \geq 0\right)$. 
So, $\frac{\partial^2 f^{\text{Var}}_{\text{d}_i}(\overline{d_i},\sigma_{\text{est}})}{\partial\sigma_{\text{est}}^2} \geq 0$ and  $f^{\text{Var}}_{\text{d}_i}(\overline{d_i},\sigma_{\text{est}})$ is a convex function. Therefore, $l(\sigma_{\text{est}})$ is a concave function. 
\begin{remark}
If function $f:\mathbb{R}^+_0 \rightarrow \mathbb{R}$  is even, concave, $f(x=0) \geq 0$, and $\exists \epsilon \left( \geq 0 \right)$ such that $f(x) \leq \epsilon $ then $x^*=\underset{x \geq 0}{\arg\min} |f(x)| \implies f(x^*) = 0$.
\end{remark}
Using the above remark, optimal value $\left(\hat{\sigma}_{\text{est}}\right)$ of the optimization problem in line 5 of Algorithm~\ref{Algo:SWLMS} can be obtained by solving $l(\hat{\sigma}_{\text{est}}) = 0$ as $l({\sigma}_{\text{est}})$ is concave, even function with
%$\left( l(\sigma_{\text{est}})=l(-\sigma_{\text{est}}) \right)$, 
$l(\sigma_{\text{est}}=0)=v$ as $f^{\text{var}}_{d} (\overline{d_i},\sigma_{\text{est}}=0) = 0$, and is upper bounded by a non-negative constant $\left( l(\sigma_{\text{est}}) \leq v, \text{ as } f^{\text{var}}_{d} (\overline{d_i},\sigma_{\text{est}}) \geq 0 \right)$. The objective function of the optimization problem under consideration i.e., $|l(\sigma_{\text{est}})|$ is neither convex nor concave when $\sigma_{\text{est}} \geq 0$ as  $|l(\sigma_{\text{est}})|$ is concave when $\sigma_{\text{est}} \in [0,\hat{\sigma}_{\text{est}}]$ and convex when  $\sigma_{\text{est}} \in \left[\hat{\sigma}_{\text{est}},\infty\right)$.    
\bibliographystyle{IEEEtran}
\bibliography{reference}

\end{document}